\documentclass[journal, letterpaper]{IEEEtran}

\usepackage{amsfonts,amsmath,amssymb}
\usepackage{dsfont}
\usepackage{cite}
\usepackage{graphicx}
 \usepackage{booktabs}

\usepackage{xcolor}

\usepackage{caption}
\usepackage{bm}
\usepackage{bbding}
\usepackage{amsmath}
\usepackage{enumerate}
\usepackage{multirow}
\usepackage[ruled]{algorithm2e}
\usepackage{algorithmicx}

\usepackage{algpseudocode}
\usepackage{algcompatible}
\usepackage{subcaption}
\usepackage{array}
\usepackage{slashbox}
\usepackage{CJK}
\usepackage{graphicx}

\usepackage{url}
\usepackage{amsthm}

  \newtheorem{proposition}{Proposition}

  \newtheorem{remark}{Remark}

  \newcommand{\figref}[1]{\figurename~\ref{#1}}

\DeclareMathOperator*{\argmax}{arg\,max}

\DeclareMathOperator{\diag}{diag}

\begin{document}
\title{Intelligent Reflecting Surface Configurations for Smart Radio Using Deep Reinforcement Learning}
\author{Wei~Wang,~\IEEEmembership{Member,~IEEE}, and~Wei~Zhang,~\IEEEmembership{Fellow,~IEEE}
\thanks{

This work was supported in part by National Key R\&D Program of China under Grant 2020YFA0711400,  Shenzhen Science \& Innovation Fund under Grant JCYJ20180507182451820, and the Australian Research Council's Project funding scheme under LP160101244.

W. Wang is with Peng Cheng Laboratory, Shenzhen, China (e-mail: wangw01@pcl.ac.cn).

W. Zhang is with School of Electrical Engineering and Telecommunications, The University of New South Wales, Sydney, NSW 2052, Australia (e-mail: w.zhang@unsw.edu.au).}
}

\maketitle

\begin{abstract}
Intelligent reflecting surface (IRS) is envisioned to change the paradigm of wireless communications from ``adapting to wireless channels" to ``changing wireless channels". However, current IRS configuration schemes, consisting of sub-channel estimation and passive beamforming in sequence, conform to the conventional model-based design philosophies and are difficult to be realized practically in the complex radio environment.  To create the smart radio environment, we propose a model-free design of IRS control that is independent of the sub-channel channel state information (CSI) and requires the minimum interaction between IRS and the wireless communication system. We firstly model the control of IRS as a Markov decision process (MDP) and  apply deep reinforcement learning (DRL) to perform real-time coarse phase control of IRS. Then, we apply extremum seeking control (ESC) as the fine phase control of IRS.  Finally, by updating the frame structure, we integrate DRL and ESC in the model-free control of IRS to improve its adaptivity to different channel dynamics. Numerical results show the superiority of our proposed  joint DRL and ESC scheme and verify its effectiveness in model-free IRS control without sub-channel CSI.
\end{abstract}


\section{Introduction}
Metasurfaces, which consist of artificially periodic or quasi-periodic structures with sub-wavelength scales, are a new design of functional materials \cite{engheta2006metamaterials, cui2014coding}. Some extraordinary electromagnetic properties observed on metasurfaces, e.g., negative permittivity and permeability, reveal its potential in tailoring electromagnetic waves in a wide frequency range, from microwave to visible light \cite{schurig2006metamaterial,liang2015anomalous, yu2011light,ChiIRS}. Intelligent reflecting surface (IRS), a.k.a. reconfigurable intelligent surface (RIS), is a type of programmable metasurfaces that is capable of electronically tuning electromagnetic wave by incorporating active components into each unit cell of metasurfaces \cite{di2019smart, wu2021intelligent, IRSch1, wang2021joint, CSirs, wu2019intelligent}. The advent of IRS is envisioned to revolutionize many industries, a major one of which is wireless communications.

Wireless communications are subject to the time-varying radio propagation environment. The effects of free space path loss, signal absorption, reflections, refractions, and diffractions caused by physical objects during the propagation of electromagnetic waves jointly render wireless channels highly dynamic \cite{tse2005fundamentals, qi2021acquisition, liu2021atmospheric}. IRS's capability of manipulating electromagnetic waves in a real time manner brings infinite possibilities to wireless communications and makes it possible for human beings to transform the design paradigm of wireless communications from ``adapting to wireless channels" to ``changing wireless channels"  \cite{di2019smart}. To this end, great research efforts have been spent to acquire the channel state information (CSI) of the sub-channels, i.e., the channels between wireless transceivers and IRS, which is widely regarded as the prerequisite of IRS reflection pattern  (passive beamforming) design \cite{wu2021intelligent}. Owing to the passive nature, IRS is unable to sense the incident signal, thereby rendering the estimation process far more complicated than traditional wireless communication systems.  In \cite{IRSch1}, a channel estimation scheme with reduced training overhead is proposed by exploiting the inter-user channel correlation. In \cite{wang2021joint,CSirs}, compressed sensing based channel estimation methods are proposed to estimate the channel responses between base station, IRS and a single-antenna user at mmWave frequency band. As the proposed schemes are focusing on a single-antenna user, their extension to multiple users with array antenna might increase multi-fold the training overhead. In \cite{WeiWang2021joint}, a joint beam training and positioning scheme is proposed to estimate the parameters of the line-of-sight (LoS) paths for IRS assisted mmWave communications. The proposed random beamforming in the training stage is performed in a broadcasting manner and thus the training overhead is independent of the user number.

Despite the aforementioned endeavors to advance the CSI acquisition techniques in IRS assisted wireless communications, the practical applications of IRS are still confronting various challenges. Firstly, channel estimations for IRS assisted wireless communications demand a radical update of the existed protocols to incorporate the coordination of the transmitter, the receiver, and the IRS. It indicates that the existing wireless systems, e.g., Wi-Fi, 4G-LTE and 5G-NR, are unable to readily embrace IRS. Secondly,  even if the perfect CSI is available, the real-time optimization of IRS reflection coefficients using convex and non-convex optimization techniques is computationally prohibitive \cite{wu2019intelligent}. Thirdly,  current solutions to IRS control, which consists of CSI acquisition and IRS reflection  designs, are based on the accurate modelling of IRS.  However, as a type of low-cost reflective metasurfaces, IRS changes its reflection coefficient via tuning the impedance, the exact value of which is dependant on the carrier frequency of the incident signal \cite{yang2016programmable, tang2019wireless}. The carrier frequency can be shifted by Doppler effects and might also vary over different users, thus the mathematical modelling of IRS in the complex radio propagation environment is inherently difficult.

To tackle the aforementioned challenges, we follow the design paradigm of model-free control by treating the wireless communication system as a (semi) black box with uncertain parameters and to optimize  reflection coefficients of IRS through deep reinforcement learning (DRL) and extremum seeking control (ESC). Compared with the prevailing designs of IRS assisted wireless communication systems \cite{CSirs,JointActivePassive,yang2019irs,yu2021irs, zhi2021ergodic, han2019large, zhi2021ris, zhi2021two, zhao2020intelligent}, our proposed scheme is model-free. Specifically, the instantaneous (or statistical) CSI of sub-channels (i.e., Tx-Rx channel, Tx-IRS channel, and IRS-Rx channel) that constitutes the equivalent wireless channel is not required. Our design, in a true sense, treats IRS as a part of the wireless channel and requires the minimum interaction with wireless communication systems. The disentanglement of IRS configuration from wireless communication system means the  improved independence of IRS and will speed up the rollout of IRS  in the future.  There are already some attempts towards the standalone operation of IRS \cite{taha2020deep, taha2021enabling, sheen2021deep}.  In \cite{taha2020deep, taha2021enabling}, deep learning and deep reinforcement learning are applied to guide the IRS to interact with the incident signal given the knowledge of the sampled channel vectors. However, in order to obtain the CSI of the Tx-IRS and IRS-Rx sub-channels, the authors propose to install channel sensors on the IRS, which is, to some extent, against the initial role of IRS as a passive device. In \cite{sheen2021deep}, to reduce the dependance on the CSI of sub-channels, a deep learning scheme is proposed to extract the interactions between phase shifts of IRS and receiver locations. However, the training data has to be collected offline, which limits its adaptability to the more general scenarios.

In this paper, our objective is to build a model-free  IRS control scheme with a higher level understanding of the radio environment, which is able to configure the IRS reflection coefficients without the CSI of the sub-channels. To this end, we adopt a typical scenario, i.e., time-division duplexing (TDD) multi-user multiple-input-multiple-output (MIMO), as an example to perform our design.  To summarize, our contributions are as follows.

\begin{itemize}
\item We model the control of IRS as a Markov decision process (MDP) and then apply DRL, specifically, double deep Q-network (DDQN) method, to perform real-time coarse phase control of IRS. The proposed DDQN scheme outperforms the other sub-channel-CSI-independent methods, e.g., multi-armed bandit (MAB), random reflection.
\item To enhance the action of DDQN, we further apply ESC as the fine phase control of IRS. Specifically, we propose a dither-based iterative method to optimize the phase shift of IRS through trial and error. We also prove that the output of the proposed dither-based iterative method is  monotonically increasing.
\item By updating the frame structure, we integrate DRL and ESC in the model-free control of IRS. The integrated scheme is more adaptive to various channel dynamics and has the potential to achieve better performance.
\end{itemize}

Numerical results show the superiority of our proposed  DRL, ESC, and joint DRL and ESC scheme and verify their effectiveness in model-free IRS control without sub-channel CSI.

The rest of the paper is organized as follows. In Section II, we introduce the system model.  In Section III, we propose a DRL enabled model-free control of IRS. In Section IV,
we propose a dither-based iterative method to enhance the action of DRL. In Section V, we present numerical results. Finally, in Section VII, we draw the conclusion.

{\em{Notations:\quad}} Column vectors (matrices) are denoted by bold-face lower (upper) case letters, $\mathbf{x}[n]$ denotes the $n$-th element in the vector $\mathbf{x}$,  $\odot$ represents the Hadamard product, $(\cdot)^*$, $(\cdot)^T$ and $(\cdot)^{H}$  represent conjugate, transpose and  conjugate transpose operation, respectively.

\section{System Model}

In this section, we introduce the system model of model-free IRS control.

\subsection{Optimal Phase Shift Vector of IRS}

\begin{figure}[tp]{
\begin{center}{\includegraphics[ height=5cm]{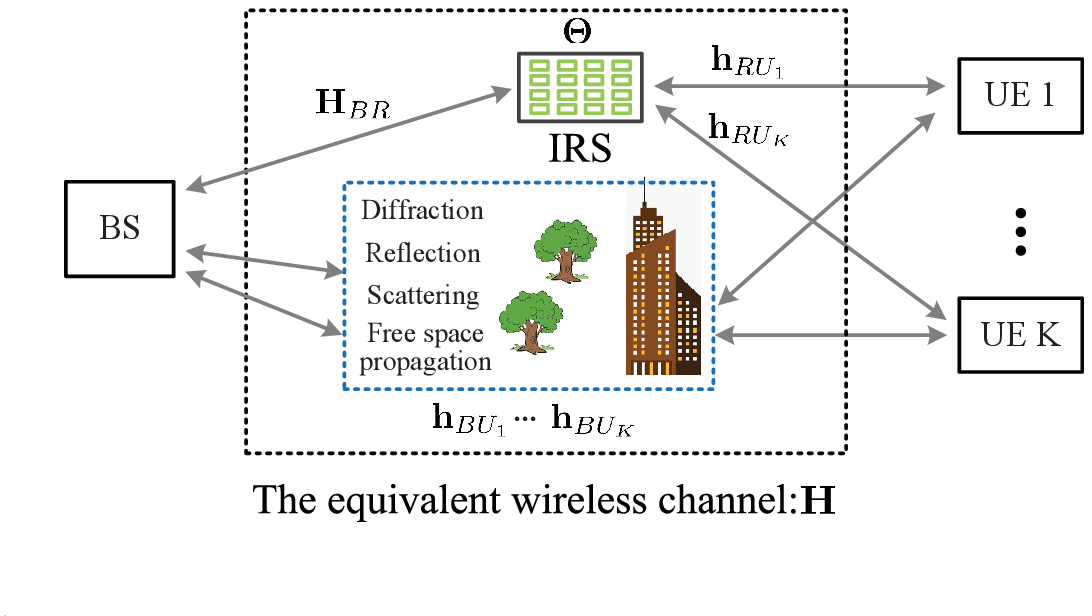}}
\caption{{ The illustration of IRS assisted wireless communications} }\label{SystemModel}
\end{center}}
\end{figure}

For IRS assisted wireless communications, the channel model between BS and a certain user $k$ can be represented as
\begin{align}
\mathbf{h}_k = \mathbf{h}_{BU_k} + \mathbf{H}_{BR}\boldsymbol{\Theta}\mathbf{h}_{RU_k}
\end{align}
where the user $k$ is equipped with a single antenna, $\mathbf{h}_{BU_k}  \in \mathbb{C}^{N_B \times 1} $ is the channel response vector between the user $k$ and BS,  $\mathbf{H}_{BR}  \in \mathbb{C}^{N_B \times N_R} $ is the channel response matrix between BS and IRS, $ \mathbf{h}_{RU_k} \in \mathbb{C}^{N_R \times 1}$ is the channel response vector between the user $k$ and IRS, $\boldsymbol{\Theta} = \diag \{\boldsymbol\theta\}$, and  $\boldsymbol\theta \in \mathbb{C}^{N_{R}\times 1} $ (with $|\boldsymbol\theta(n)|=1$) is the  phase shift vector of the IRS. 
Accordingly, the multi-user channel  is  written as
\begin{align}
\mathbf{H} = \mathbf{H}_{BU} + \mathbf{H}_{BR} \boldsymbol{\Theta} \mathbf{H}_{RU}
\end{align}
where
\begin{align}
 &\mathbf{H}= [\mathbf{h}_1, \mathbf{h}_2, \cdots, \mathbf{h}_K] \notag \\
 & \mathbf{H}_{BU}= [\mathbf{h}_{BU_1}, \mathbf{h}_{BU_2}, \cdots, \mathbf{h}_{BU_K}] \notag \\
 &\mathbf{H}_{RU} = [\mathbf{h}_{RU_1}, \mathbf{h}_{RU_2}, \cdots, \mathbf{h}_{RU_K}]\notag
\end{align}
{And the relationship between the aggregated equivalent channel $\mathbf{H}$ and the sub-channels is shown in \figref{SystemModel}.}

The objective of reinforcement learning based IRS configuration is to develop a widely compatible method that can be deployed in various scenarios of wireless communications without any knowledge of the wireless system's internal working mechanism. Mathematically, the problem is formulated as
\begin{align}  \label{Opt}
\begin{split}
&\max_{\boldsymbol{\theta}} \;\; P_m    \\
&\;s.t. \;\; \;\; \boldsymbol\theta[n] = e^{-j  \boldsymbol{\varphi}[n]},\; \forall n \in \{1, 2, \cdots, N_{R}\}  \\
&\quad \quad \;\;\; \boldsymbol{\varphi}[n] \in \mathcal{B},\;\;\;   \forall n \in \{1, 2, \cdots, N_{R}\}
\end{split}
\end{align}
where $P_m$ is the performance metric of the wireless system that is to be optimized. $P_m$ is dependent on the wireless channel $\mathbf{H}$, and $\mathbf{H}$ is dependent on the reflection pattern $\boldsymbol{\theta}$.
 $\boldsymbol{\varphi}[n]$ is the quantized phase selected from a finite set $\mathcal{B} = \left\{-\pi, \frac{-2^r+2}{2^r}\pi, \frac{-2^r+4}{2^r}\pi, \cdots,   \pi \right\} $ with $2^r+1$ possible values.

 It is worth mentioning that the model-free control does not need to know the exact relationship between the objective $P_m$ and variable $\boldsymbol{\theta}$.

\subsection{A Typical Scenario -- TDD Multi-User MIMO }

Without loss of generality, we  use a typical scenario in wireless communication, i.e., TDD multi-user MIMO, to illustrate our design philosophy. In TDD, by exploiting the channel reciprocity, the BS can estimate the downlink channel from the pilot of the uplink channel. Thus, TDD multi-user MIMO consists of two stages (refer to \figref{FS1}), i.e., \emph{uplink pilot transmission} and \emph{downlink data transmission}\cite{kim2013energy,zhang2015large}.

At uplink  stage, the pilot transmits from multiple users to BS simultaneously. The received pilot signal is represented as
\begin{align}
\mathbf{Y}_{U} & = \mathbf{H}\mathbf{S}  + \mathbf{N}
\end{align}
where  $\mathbf{S}\in \mathbb{C}^{K\times K}$ is the pilot pattern, $\mathbf{N} \in \mathbb{C}^{N_B\times K} $ is the additive white Gaussian noise. Upon receiving the pilot, BS performs  minimum mean square error (MMSE) estimation of the channel matrix, i.e.,
\begin{align}
\hat{\mathbf{H}} =  \mathbf{Y}_U \mathbf{S}^H(\mathbf{S}\mathbf{S}^H+ \sigma^2_U\mathbf{I})^{-1} \label{MMSEest}
\end{align}
When $\mathbf{S}$ is an unitary matrix, \eqref{MMSEest} is further expressed as
\begin{align}
\hat{\mathbf{H}} = \frac{ \mathbf{Y}_U \mathbf{S}^H}{1+ \sigma^2_U }  \label{MMSEest2}
\end{align}

At downlink stage, data transmission with zero-forcing (ZF) precoding is performed, and the precoding matrix is represented as
\begin{subequations}
\begin{align}
\mathbf{M} &= [\mathbf{m}_1, \mathbf{m}_2, \cdots, \mathbf{m}_K]^H   \\
&=   \mathbf{D}_{{P}}(\hat{\mathbf{H}}^H\hat{\mathbf{H}})^{-1}\hat{\mathbf{H}}^H
\end{align}
\end{subequations}
where $\mathbf{D}_p = {\rm diag}([\frac{1}{\|\mathbf{m}_1\|_2}, \frac{1}{\|\mathbf{m}_2\|_2}, \cdots, \frac{1}{\|\mathbf{m}_K\|_2}])$ is for power normalization. The received signal of user $k$ is given by
\begin{align}
y_{D,k} = \mathbf{m}_k^H \mathbf{h}_k x_k + \sum_{l \neq k}^K \mathbf{m}_l^H \mathbf{h}_k  x_l + n_k
\end{align}
where $x_k$ is the signal intended to user $k$ ($\mathbb{E}(x_k)=0$ and  $\mathbb{E}(|x_k|^2) = 1$, $\forall k \in\{1, \cdots, K\}$), and $n_k \sim \mathcal{CN}(0, \sigma_k^2)$ is the additive white Gaussian noise. Thus, the  signal-to-noise ratio of the $k$-th user is
\begin{align}
SINR_k = \frac{|\mathbf{m}_k^H\mathbf{h}_k|^2}{\sum_{l \neq k}^K |\mathbf{m}_l^H \mathbf{h}_k|^2 + \sigma^2_k }
\end{align}

For a communication system, the performance metrics can be SINR, data rate, frame error rate (FER), and etc. Without loss of generality, we adopt the sum data rate as the performance metric, i.e.,
\begin{align}
 P_m = \sum_{k=1}^K r_k  = \sum_{k=1}^K \log_2 (1 + SINR_k) \label{Performance}
\end{align}

\subsection{Channel Model}

We assume the Rician channel model for $\mathbf{h}_{BU_k}$, $\mathbf{H}_{BR}$ and $\mathbf{h}_{RU_k}$. Take $\mathbf{H}_{BR}$ as an example, it is represented as
\begin{align}
\mathbf{H}_{BR}= \sqrt{\frac{K}{K+1}}\mathbf{H}_{BR, LoS} + \sqrt{\frac{1}{K+1}}\mathbf{H}_{BR, NLoS}
\end{align}
where $\mathbf{H}_{BR, LoS}$ denotes the deterministic LoS component, $\mathbf{H}_{BR, NLoS}$ denotes the fast fading NLoS component, and the component of which is independent and identically distributed (i.i.d.) circularly symmetric complex Gaussian random variables with zero-mean and unit variance, and $K$ is the ratio between the power in the LoS path and the power in the NLoS paths \cite{goldsmith2005wireless}.

The LoS component is position-dependent and is thus slow-time-varying; The NLoS components are caused by the multi-path effects and are thus fast-time-varying \cite{ChannelVary}. Combining the characteristics of wireless channel with the setting of reinforcement learning, we introduce the following two concepts.

(1) \textbf{Channel block:}  One channel block consists of the uplink pilot transmission stage and downlink data transmission stage (as shown in \figref{FS1}), and the channel matrix is constant during the channel block.

(2) \textbf{Channel episode:}  One channel episode consists of $T$ channel blocks (as shown in \figref{FS1}). The LoS component within one channel episode remains constant; The NLoS components change over time, and the NLoS components of different channel blocks are i.i.d.

\begin{figure}[tp]{
\begin{center}{\includegraphics[ height=3.3cm]{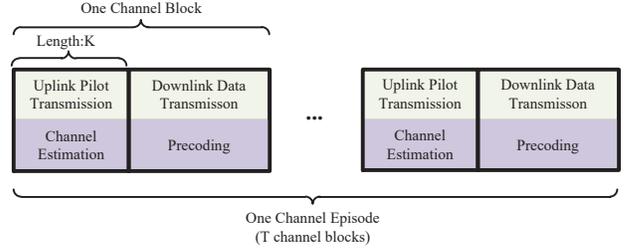}}
\caption{The frame structure of the typical TDD multi-user MIMO (green part) and the corresponding signal processing procedures (purple part) }\label{FS1}
\end{center}}
\end{figure}

\section{Model-Free IRS Control Enabled By Deep Reinforcement Learning}

In this section, we apply DRL to model-free IRS control.

\subsection{Design Objectives}

We aim to achieve stand-alone operation of the wireless communication system and the IRS, and our design includes the following characteristics.

\begin{itemize}
\item
\textbf{Wireless Communication System:} The wireless communication system is almost unaware of the existence of the IRS, except that it needs to feed back its instantaneous performance to the IRS controller. In this regard, the uplink pilot transmission and downlink data transmission exactly follow the conventional structure in \figref{FS1}.

\item
\textbf{IRS:} The IRS is strictly regarded as part of the wireless channel and will not be jointly designed with the wireless communication system.  The configuration of IRS is based on (a) the performance feedback from the wireless system and (b) its learned policy through trail-and-error interaction with the dynamic environment. And the IRS is unaware of the working mechanism of the wireless communication system.
\end{itemize}

A salient advantage of the proposed design is that the IRS can be deployed in various wireless communication applications, e.g., Wi-Fi, 4G-LTE, 5G-NR, without updating their existing protocols, which will speed up the roll-out of IRSs. Another benefit is that, by treating the existing wireless communication system as a black box, the configuration of IRS does not require the overhead-demanding channel sounding process to acquire the CSI of the subchannels,  i.e., $\mathbf{H}_{BU}$, $\mathbf{H}_{BR}$, and $\mathbf{H}_{RU}$, that constitute the aggregated equivalent channel $\mathbf{H}$.

Our design is primarily based on the \emph{reinforcement learning} technology. Specifically, the IRS and its controller are the \emph{agent}, the wireless communication system, which comprises transmitter, wireless channel, and receiver, is the \emph{environment}.  The relationship between the different parties is given in \figref{DRL}. Initially, the agent takes  random \emph{actions} and the environment responds to those actions by giving rise to rewards and presenting new situations to the agent \cite{sutton2018reinforcement,kaelbling1996reinforcement}. Through trail-and-error interaction with the wireless communication system, the agent gradually learns the optimal policy to maximize the expected return over time. In this regard, IRS, which is capable of changing the radio environment, is analogous to the human body, and the reinforcement learning method, which guides the action of IRS, is analogous to the human brain. The integration of the IRS and the reinforcement learning method is the pathway to creating the smart radio environment.

\begin{figure}[tp]{
\begin{center}{\includegraphics[ height=7cm]{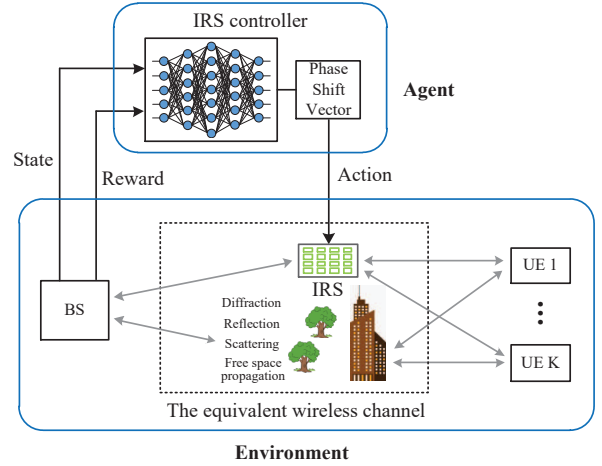}}
\caption{The structure of model-free IRS configuration enabled by deep reinforcement learning}\label{DRL}
\end{center}}
\end{figure}

\subsection{Basics of Deep Reinforcement Learning}
To facilitate the presentation of our design, we briefly introduce some key concepts of DRL in this subsection.

\subsubsection{Objective of Reinforcement Learning}

An MDP is specified by $4$-tuple $\langle \mathcal{S},  \mathcal{A},  {P},  {R}\rangle $, where $\mathcal{S}$ is the state space, $\mathcal{A}$ is the action space, ${P}$ is the state transition probability, and ${R}$ is the immediate reward received by the agent. When an agent in the state  $s\in {S}$ takes the action $a\in \mathcal{A}$, the environment will evolve to the next state $s^{\prime}\in \mathcal{S}$ with probability $ P(s^{\prime}|s,a) = \rm{Pr}(S_{t+1} = s^{\prime}| S_t = s, A_t = a) $, and in the meantime, the agent will receive the immediate reward $R_{s \rightarrow s^{\prime}}^a$. Adding the time index to $S,A,R$, the evolution of an MDP can be represented using the following trajectory
\begin{align}
\langle S_0, A_0, R_1, S_1, A_1, R_2, \cdots, S_{T-1}, A_{T-1}, R_T, S_T, \cdots \rangle \label{trajectory}
\end{align}

The agent's action is directed by the policy function
\begin{align}
\pi(a|s) = \rm{Pr}(A_t=a|S_t=s)
\end{align}
which is the probability that the agent takes action $a$ when the current state is $s$.  A reinforcement learning task intends to find a policy that achieves a good return over the long run, where the return is defined as the cumulative discounted future reward, i.e.,
\begin{align}
U_t  = \Sigma_{\tau = 0}^{\infty} \gamma^{\tau}R_{t+\tau+1} \notag
\end{align}
and $\gamma \in [0, 1]$ is the discount factor for future rewards. Owing to the randomness of state transition (caused by the dynamic environment) and action selection, the return $U_t$ is a random variable. Mathematically,  the agent's goal in reinforcement learning is to find a good policy that maximizes the expected return, i.e.,
\begin{align}
\max_{\pi} \;\; \mathbb{E}(U_t)
\end{align}

\subsubsection{Action-Value Function and Optimal Policy}

One key metric for action selection in reinforcement learning is action-value function, i.e.,
\begin{align}
Q_{\pi}(s, a) = \mathbb{E}[U_t|S_t = s, A_t = a]
\end{align}
which is the conditional expected return for an agent to pick action $a$ in the state $s$ under the policy $\pi$.  For any policy $\pi$ and any state $s \in \mathcal{S}$, action-value function satisfies the following recursive relationship, i.e.,
\begin{align}
& Q_{\pi}(s,a) = \mathbb{E}_{s^{\prime}}\left[R_{s \rightarrow s^{\prime}}^a + \gamma \sum_{a^{\prime}\in \mathcal{A}} \pi(a^{\prime}|s^{\prime})Q_{\pi}(s^{\prime}, a^{\prime})\Big|s^{\prime}, a^{\prime}\right] \notag \\
&= \sum_{s^{\prime}\in\mathcal{S}}P(s^{\prime}|s,a) \left(R_{s \rightarrow s^{\prime}}^a + \gamma \sum_{a^{\prime}\in \mathcal{A}} \pi(a^{\prime}|s^{\prime})Q_{\pi}(s^{\prime}, a^{\prime}) \right) \label{Bellman}
\end{align}
where $R_{s \rightarrow s^{\prime}}^a$ is the immediate reward when the environment transits from state $s$ to state $s^{\prime}$ after taking the action $a$, and Eq. \eqref{Bellman} is the well-known Bellman equation of action-value function \cite{sutton2018reinforcement}.

A policy is defined to be better than another policy if its expected return is greater than that of for all states and all actions. Thus, the optimal action-value function is
\begin{align}
Q^{*}(s, a) := & Q_{\pi^{*}}(s, a) \notag \\
 = &\max_{\pi} Q_{\pi}(s, a),\;\; \forall s \in \mathcal{S}, a \in \mathcal{A} \label{optQ}
\end{align}
With the optimal action-value function, the optimal policy is obviously
\begin{align}
\pi^{*}(a|s) = \left\{\begin{array}{cc}
                     1, &  {\rm if} \; a = \argmax_{a\in \mathcal{A}}  Q^{*}(s, a)\\
                     0, &  {\rm otherwise} \qquad \qquad \qquad \quad \;\;\;
                   \end{array}   \right. \label{optP}
\end{align}
Combining \eqref{optQ}, \eqref{optP} with \eqref{Bellman}, the Bellman optimality equation for $Q^{*}(s,a)$ is given by
\begin{align}
Q^{*}(s,a) & =  \sum_{s^{\prime}\in\mathcal{S}}P(s^{\prime}|s,a) (R_{s \rightarrow s^{\prime}}^a + \gamma\max_{a^{\prime}\in \mathcal{A}} Q^{*}(s^{\prime}, a^{\prime})) \notag \\
& = \mathbb{E}_{s^{\prime}}[ R_{s \rightarrow s^{\prime}}^a + \gamma \max_{a^{\prime}\in \mathcal{A}} Q^{*}(s^{\prime}, a^{\prime})|s_t=s, a_t =a] \label{Bellman2}
\end{align}

With the Bellman optimality equation, the optimal policy $\pi^{*}(a|s)$ or the optimal action-value function $Q^{*}(s,a)$ can be obtained via iterative methods, i.e., policy iteration based methods and value iteration based methods\cite{sutton2018reinforcement}. Hereinafter, we will mainly focus on value iteration based methods.

\subsubsection{Temporal Difference Learning}

The aforementioned iterative methods require the complete knowledge of the environment, i.e., state transition probability $p(s^{\prime}|s,a)$, reward function $R_{s \rightarrow s^{\prime}}^a$, etc.  However, the explicit knowledge of environment dynamics is unavailable in practice. The conditional expectation in \eqref{Bellman2} can be realized via numerically averaging over the sample sequences of states, actions, and rewards from
actual interaction with the environment, e.g., temporal difference method, or Monte Carlo method.

Upon observing a new segment of the trajectory in \eqref{trajectory}, i.e., $\langle S_t = s, A_t = a, R_{t+1} = R_{s \rightarrow s^{\prime}}^a, S_{t+1} = s^{\prime} \rangle$,
the action-value function $Q(s,a)$ updates as follows:
\begin{align}
& Q_{t+1}(s,a) =   \notag \\
&  Q_t(s,a) + \alpha \left(R_{s \rightarrow s^{\prime}}^a + \gamma \max_{a^{\prime}\in \mathcal{A}} Q_t(s^\prime, a^\prime) - Q_t(s, a)\right) \label{Qiterative}
\end{align}
where $\alpha \in (0,1]$ is the learning rate and the following term inside the bracket is the error between the estimated Q value and the return. It means that the value function is updated in the direction of the error, and iteration will terminate when the error becomes infinitesimal.

\subsubsection{Double Deep Q-Network (DDQN)}

When the state $s$ and the action $a$ are both discrete, the optimal state-action function $Q^*(s,a)$ can be obtained as a lookup table, which is also known as Q-table \cite{arulkumaran2017deep}, following the iterative procedures in \eqref{Qiterative}. However,  the size of the state (or action) space can be prohibitively large, and the state (or action) can  even be continuous.
In such cases, it is impractical to represent $Q(s,a)$ as a lookup table. Fortunately, the deep neural network (DNN) can be adopted to approximate the Q-table as $Q(s,a) \approx \widetilde{Q}(s,a; \mathbf{w})$,
which enables reinforcement learning to scale to more generalized decision-making problems. The coefficients $\mathbf{w}$ of $\widetilde{Q}(s,a; \mathbf{w})$ are the weights of the DNN, and the DNN is termed as deep Q-network (DQN) \cite{mnih2015human, arulkumaran2017deep}.

The trajectory segment $\langle S_t = s, A_t = a, R_{t+1} = R_{s \rightarrow s^{\prime}}^a, S_{t+1} = s^{\prime} \rangle$  in  \eqref{trajectory} constitutes an ``experience sample" that will be used to train the DQN, and in accordance with \eqref{Qiterative}, the loss function adopted during the training process of DQN is
\begin{align}
Loss = \Big( \underbrace{R_{s \rightarrow s^{\prime}}^a + \gamma \max_{a^{\prime} \in \mathcal{A}} \widetilde{Q}(s^\prime, a^\prime; \mathbf{w})}_{T_{DQN}} - \widetilde{Q}(s, a; \mathbf{w}) \Big)^2
\end{align}
where $T_{DQN}$ is the target value fed to the network.

The target $T_{DQN}$ is dependent on the immediate reward $R_{s \rightarrow s^{\prime}}^a $, as well as the output of the DQN $\widetilde{Q}(s^\prime, a^\prime; \mathbf{w})$. Such structure will inevitably result in over-estimation of the action state value (a.k.a., Q value) during the training process and thus will significantly degrade the performance of DRL. To mitigate over-estimation, we will adopt the double DQN (DDQN) structure  \cite{mnih2013playing, van2016deep} in our design.

The fundamental idea of DDQN is to apply a separate target network  $\widetilde{Q}(s^\prime, a^\prime; \mathbf{w}^{-})$ to estimate the target value \cite{van2016deep}, and the expression of target in DDQN is
\begin{align}
T_{DQN} = R_{s \rightarrow s^{\prime}}^a + \gamma \widetilde{Q}(s^\prime, {\argmax_{a^{\prime} \in \mathcal{A}} \widetilde{Q}(s^\prime, a^\prime; \mathbf{w});  \mathbf{w}^{-}}) \label{TargetVal}
\end{align}
To summarize, DDQN differs from DQN in the following two aspects, i.e., (1) the optimal action is selected using the DQN $\widetilde{Q}(s^\prime, a^\prime; \mathbf{w})$
whose weights are $\mathbf{w}$, and (2) the Q value of the target value is taken from the target network whose weights are $\mathbf{w}^{-}$.

\subsection{Model-Free Control of IRS Using Deep Reinforcement Learning}

To apply reinforcement learning to model-free IRS configuration, we firstly model IRS assisted wireless communications as an MDP.

\begin{itemize}
\item \emph{Agent}: The agent is IRS controller, which is capable of autonomously interacting with the environment via IRS to meet the design objectives.

\item \emph{Environment}: The environment refers to the things that the agent interact with, which includes BS, wireless channel, IRS, and mobile users.

\item \emph{State}: To facilitate the  accurate prediction of expected next rewards and next states given an action,  we define the state as $\{\mathbf{H},  \boldsymbol{\theta}\}$,  which  consists of two sub-states, namely the equivalent wireless channel $\mathbf{H}$ and the reflection vector $\boldsymbol{\theta}$ of IRS.

\item \emph{Action}: The action is defined as the incremental phase shift of the current reflection pattern, i.e.,
\begin{align}
 \boldsymbol{\theta}^{(t+1)} = \boldsymbol{\theta}^{(t)} \odot \Delta \boldsymbol{\theta}^{(t)}
\end{align}
where $\odot$ is the Hadamard (element-wise) product, $\boldsymbol{\theta}^{(t)}$ is the reflection pattern at the $t$-th channel block, and $\Delta \boldsymbol{\theta}^{(t)} $ is the incremental phase shift of  $\boldsymbol{\theta}^{(t)}$.
We use the subset (or full set) of the discrete Fourier transform (DFT) vectors as the action set. For example, when the size of action space is $5$, we set  $\mathcal{A} = \left\{\mathbf{v}(-\frac{6}{N_R}), \mathbf{v}(-\frac{2}{N_R}), \mathbf{v}(0), \mathbf{v}(\frac{2}{N_R}), \mathbf{v}(\frac{6}{N_R})\right\}$, where $\mathbf{v}(\Psi_R)$ is the steering vector \footnotemark, i.e.,
\begin{align}
\mathbf{v}(\Psi_R) &= \left[1,\; e^{j  \pi \Psi_R},\; \cdots,\; e^{j  (N_{R} - 1)\pi \Psi_R} \right]^T \notag
\end{align}
When $\Delta\boldsymbol{\theta}^{(t)}  = \mathbf{v}(0)$, the sub-state $\boldsymbol{\theta}$ stays unchanged, and the sub-state $\mathbf{H}$ changes merely due to the variation of NLoS components; $\Delta\boldsymbol{\theta}^{(t)}  = \mathbf{v}(-\frac{2}{N_R})$ and $\Delta\boldsymbol{\theta}^{(t)}  = \mathbf{v}(\frac{2}{N_R})$ are towards the opposite directions, which enables the agent to quickly correct from a negative action; $\Delta\boldsymbol{\theta}^{(t)}  = \mathbf{v}(-\frac{6}{N_R})$ and $\Delta\boldsymbol{\theta}^{(t)}  = \mathbf{v}(\frac{6}{N_R})$ are used to speed up the transition of reflection pattern.

\item \emph{Reward}: The immediate reward after transition from  $s$ to $s^{\prime}$ with action $a$ is defined as
\begin{align}
R=\left\{  \begin{array}{cc}
             P_m, & \; {\rm when} \;\; P_m \geq P_{th} \\
             P_m-100, &\; {\rm when} \;\; P_m < P_{th}
           \end{array}
\right.
\end{align}
where $P_{th}$ is a performance threshold. When $P_m$ is less than  $P_{th}$, we add an penalty $-100$ to encourage the IRS to maximize performance, while maintaining an acceptable performance above the threshold.

\end{itemize}

\footnotetext{Without loss of generality, we assume that the reflector array of IRS is a uniform linear array (ULA).}

\begin{remark}{\rm
The reasons for using incremental phase shift, rather than the absolute phase shift, as the action are two-fold. On one hand,  we need to build the Markov property of the state transmission, and, on the other hand, we intend to reduce the size action space and accelerate convergence rate. }
\end{remark}

\begin{algorithm}[h]

    \caption{Double DQN based model-free  IRS control for IRS-assisted wireless communications}
           Initialize parameters $s_0, \epsilon$;\\
           Initialize the FIFO memory $\mathcal{M}$ with the size $N_m$; \\
           Initialize the weights of the DQN $\mathbf{w}$ and set the target network as $\mathbf{w}^{-} = \mathbf{w}$ \\
           \textbf{for} $t = 0, 1, 2,  \cdots$ \textbf{do} \\
           \quad Input $s_t$ to the DQN and obtain the state-action values $\widetilde{Q}(s_t, a; \mathbf{w}), a \in \mathcal{A}$;\\
           \quad With $\widetilde{Q}(s_t, a; \mathbf{w}), a \in \mathcal{A}$, select an action $a_t$ using $\epsilon$-greedy policy;\\
           \quad Receive the reward $r_{t+1}$ and the estimated channel response $\hat{\mathbf{H}}_{t+1}$,  and compute the next state $s_{t+1}$ from $\hat{\mathbf{H}}_{t+1}$, $s_t$ and $a_t$.\\
           \quad Store the experience tuple $\langle s_t, a_t, r_{t+1}, s_{t+1}  \rangle$ to the FIFO memory $\mathcal{M}$;

           \quad \textbf{If}  $|\mathcal{M}|\geq N_{e}$ \\

           \qquad Randomly select a mini batch of $N_{e}$ experience \\
            \quad tuples $\langle s_i, a_i, r_{i+1}, s_{i+1}  \rangle$ from $\mathcal{M}$. \\
           \qquad Calculate the target values $T_{DQN, i}$ for the mini  \\
           \quad batch according to \eqref{TargetVal}. \\
           \qquad With the input $\{ s_i\}$ and the output $\{T_{DQN, i}\}$, train \\
           \quad the DQN, and update its weights $\mathbf{w}$. \\
           \qquad \textbf{If} $t \mod N_{TNet} = 0$, update the weights of the \\
           \quad target network, i.e., set $\mathbf{w}^{-} = \mathbf{w}$. \\
           \qquad \textbf{end if} \\
           \quad \textbf{end if}\\
    \textbf{       end for }
\end{algorithm}

Based on the modeled MDP and the basics of DRL presented in Subsection B, we propose to maximize the expected return (cumulative discounted future reward) using Algorithm 1. Some of the key techniques applied in Algorithm 1 are explained as follows.

\begin{figure}[tp]{
\begin{center}{\includegraphics[ height=6cm]{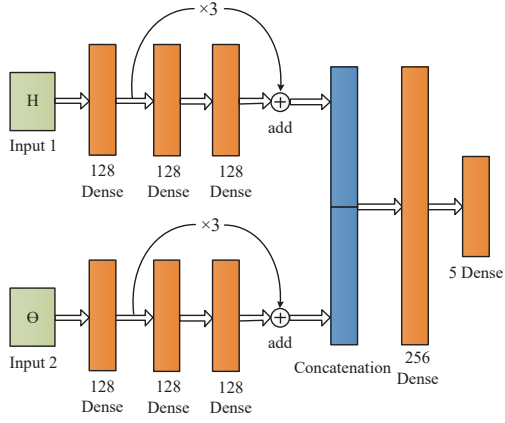}}
\caption{Structure of the DQN}\label{NetworkStructureAAA}
\end{center}}
\end{figure}

\subsubsection{DDQN}  Different from the naive DQN method, where the DQN $\widetilde{Q}(s, a; \mathbf{w})$ (with the weights $\mathbf{w}$) is used to generate the target value, we use a separate target network $\widetilde{Q}(s, a; \mathbf{w}^-)$ (with the weights $\mathbf{w}^-$) to generate the target value, and the weights of the target network are updated  by $\mathbf{w}^- = \mathbf{w}$ in every $N_{TNet}$ time intervals. The structure of the network is shown in \figref{NetworkStructureAAA}. Specifically, we apply the deep residual network (ResNet) \cite{he2016deep} to process two sub-states (i.e., $\mathbf{H}$ and $\boldsymbol{\theta}$), and then we fuse the processed information of the two sub-states using a two-layer dense network. The activation function that we use is the Swish function \cite{ramachandran2017searching}.

\subsubsection{$\epsilon$-Greedy Policy}  Given the perfect $\widetilde{Q}(s, a; \mathbf{w})$, the optimal policy is to select the action that yields the largest state-action value. However, the perfect $\widetilde{Q}(s, a; \mathbf{w})$ demands for a infinite size of experiences, which is impractical and infeasible in the dynamic wireless environment. Therefore, it is necessary for the agent to keep exploring to avoid get stuck with a sub-optimal policy. To this end, we apply an $\epsilon$-greedy policy. In $\epsilon$-greedy policy, $\epsilon$ refers to the probability of choosing to explore, i.e., randomly select from all the possible actions, and $1-\epsilon$ is the probability of choosing to exploit the obtained DQN in decision making. In this regard, the $\epsilon$-greedy policy is represented as
\begin{align}
\pi^{\epsilon} = \left\{\begin{array}{cc}
                     \pi^*(a/s), & w.p. \;\; 1- \epsilon\\
                     P(a) =  \frac{1}{|\mathcal{A}|}, &  w.p. \; \;  \epsilon
                   \end{array}   \right.
\end{align}
where $\pi^*(a/s)$ as the policy based on the Q-network, which is introduced in \eqref{optP}. In our design, $\epsilon$ is initially set to $1$ and decreases exponentially at a rate of $\vartheta, (0<\vartheta<1)$ every time interval until its reaches the lower bound $\epsilon_{min}$.

\subsubsection{Experience Replay} Instead of training the DQN with the latest experience tuple, we store $N_e$ recent experience tuples in the memory $\mathcal{M}$ in ``first in, first out" (FIFO) manner, i.e., queue data structure, and then randomly fetch a mini-batch of $N_e$ experience samples from $\mathcal{M}$ to train the DQN.

\begin{figure}[tp]{
\begin{center}{\includegraphics[ height=6cm]{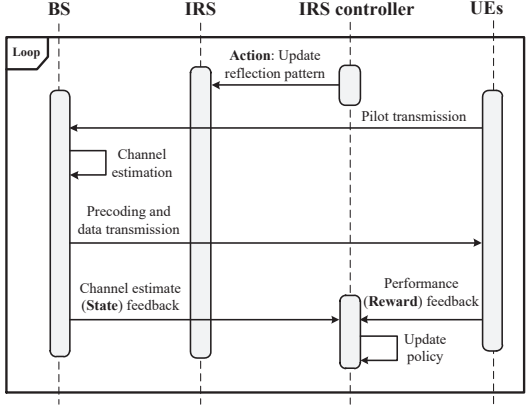}}
\caption{Sequence diagram of the proposed model-free IRS configuration}\label{UML}
\end{center}}
\end{figure}
\subsection{Summarizing The Work Flow of Model-Free IRS Configuration  }

In this subsection, we summarize the work flow of the proposed model-free IRS configuration. To this end, we plot the sequence diagram in \figref{UML}.

According to \figref{UML}, in a specific loop, the IRS is configured with a reflection pattern ($  \boldsymbol{\theta}^{(t+1)} = \boldsymbol{\theta}^{(t)} \odot \Delta \boldsymbol{\theta}^{(t)}$) according to the $\epsilon$-greedy policy, and then UEs and BS perform uplink pilot transmission and downlink data transmission sequentially as if there exists no IRS. After that, BS sends the estimated channel matrix ($\hat{\mathbf{H}}^{(t+1)}$) to IRS controller, which can be fulfilled through wired communications, and UEs send back their performance metrics to the IRS controller. Finally, according to the received channel estimate ($\hat{\mathbf{H}}^{(t+1)}$) and performance feedback ($P_m^{t+1}$), IRS controller derives the tuple $\langle \{\hat{\mathbf{H}}^{(t)}, \boldsymbol{\theta}^{(t)} \}, \Delta\boldsymbol{\theta}^{(t)}, R^{(t+1)},    \{\hat{\mathbf{H}}^{(t+1)}, \boldsymbol{\theta}^{(t+1)} \}\rangle$ and stores it in the FIFO queue as the training data for the DQN $\widetilde{Q}(s_t, a; \mathbf{w})$.

Compared with the traditional TDD multi-user MIMO, the extra efforts of incorporating IRS are merely the feedback of $\hat{\mathbf{H}}^{(t)}$ and $P_m^{t+1}$. The former can be easily achieved via wired communications between BS and IRS, and the latter costs negligible wireless communication resources of the mobile UEs. It is also noteworthy that IRS controller is unaware of the working mechanism of BS and UEs and does not require the CSI of the sub-channels.

\section{Enhancing IRS Control Using Extremum Seeking Control}

In DRL, action space is restrained for a fast convergence rate, which limits the phase freedom of IRS. To enhance the control of IRS, another model-free real-time optimization method, namely, extremum seeking control  (ESC), is used to design the fine phase control of IRS.

\subsection{Model-Free Control of IRS Using ESC}

ESC is model-free method to realize a learning-based adaptive controller for maximizing/minimizing certain system performance metrics \cite{ariyur2003real, ESCIEEE}. The first application of ESC can be traced back to the work of the French engineer Leblanc  in 1922 to maintain an efficient power transfer for a tram car  \cite{TanESC}. The basic idea of ESC is to add a dither signal (e.g., sinusoidal signal \cite{wang1999optimizing, nevsic2009extremum}, and random noise \cite{carnevale2010maximizing}) to the system input and observing its effect on the output  to obtain an approximate \emph{implicit gradient} of a nonlinear static map of the unknown system \cite{atta2015extremum, ariyur2003real}.

According to the design philosophy of ESC, we propose a dither-based model-free control of IRS as in \figref{ESC}. Our design consists of three parts, i.e., dither signal generation module, ascent direction estimation module, and parameter update module. Dither signal generation module generates random dither/pertubation signal to probe the response  $P_m(\cdot)$ of the  system; gradient estimation module determines the update direction of the system input according to the system performance  $P_m(\boldsymbol{\varphi} + \Delta\boldsymbol{\varphi} )$  to guarantee the monotonic increase of the performance, and it also guides on-off switch of the random dither signal generation; parameter update module updates the system input $\boldsymbol{\varphi}$ according to the estimated direction.

\begin{figure}[tp]{
\begin{center}{\includegraphics[ height=5.5cm]{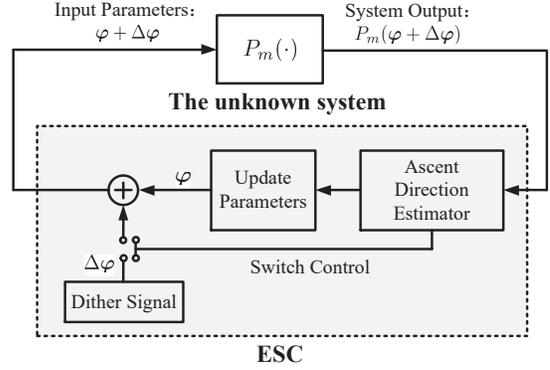}}
\caption{Principle of the proposed ESC-inspired dither-based iterative method}\label{ESC}
\end{center}}
\end{figure}

Specifically, the iterative process in \figref{ESC} runs as follows.

\vspace{0.1cm}
\noindent \textbf{\emph{Step 1. Dither Signal Generation}}
\vspace{0.1cm}

Generate a small random dither signal through uniform random distribution\footnotemark, i.e.,
\begin{align}
\Delta \boldsymbol{\varphi}[n]  = \frac{a}{2^{r-1}},  \;\;\; a \in \mathcal{U}\left\{-\frac{2^{r-1}}{N_R}, \frac{2^{r-1}}{N_R} \right\}
\end{align}

Then, add the dither signal  $\Delta \boldsymbol{\varphi} $ to the parameter $\boldsymbol{\varphi}$, use $ \boldsymbol{\varphi} + \Delta \boldsymbol{\varphi}$ as the input of the system, and receive the feedback of the performance metric $P_m(\boldsymbol{\varphi} + \Delta \boldsymbol{\varphi})$.

\footnotetext{The parameter selection will be explained in the following context.}

\vspace{0.1cm}
\noindent \textbf{\emph{Step 2. Direction Estimation and Parameter Update}}

\vspace{0.1cm}
\noindent{\underline{Condition 1}}. If  $P_m(\boldsymbol{\varphi} + \Delta \boldsymbol{\varphi})\geq P_m(\boldsymbol{\varphi})$, adopt $\Delta \boldsymbol{\varphi}$ as the direction.  Update the parameter as
\begin{align}
\boldsymbol{\varphi}\leftarrow\boldsymbol{\varphi}+ \Delta \boldsymbol{\varphi}, \label{Opt1}
\end{align}
and update the performance metric as
\begin{align}
P_m(\boldsymbol{\varphi}) \leftarrow P_m(\boldsymbol{\varphi}+\Delta \boldsymbol{\varphi})
\end{align}
Then, jump to Step 1 for the next iteration;

\vspace{0.15cm}
\noindent{\underline{Condition 2}}.  Else if $P_m(\boldsymbol{\varphi} + \Delta \boldsymbol{\varphi})<P_m(\boldsymbol{\varphi})$, adopt $-\Delta \boldsymbol{\varphi}$ as the direction. Update the parameter as
set
\begin{align}
\begin{split}
\boldsymbol{\varphi}& \leftarrow  \boldsymbol{\varphi} - \Delta \boldsymbol{\varphi}   \label{Opt2}
\end{split}
\end{align}
Turn off the dither signal, use only $\boldsymbol{\varphi}$ as the system input, and measure the system performance $P_m(\boldsymbol{\varphi})$. Then, jump to Step 1 for the next iteration.

\vspace{0.1cm}

It is noteworthy that each iteration uses one or two time intervals, and each iteration can guarantee the monotonic increase of the performance metric $P_m$, which is validated in the following proposition.

\begin{proposition} {\rm
Each iteration in ESC based iterative process can guarantee the monotonic increase of the performance metric $P_m$ given that the norm of the random dither signal, i.e.,  $\| \Delta \boldsymbol{\varphi} \|$, is small enough.}
\end{proposition}

\begin{proof}
To prove Proposition 1, it is essential to validate that the operation \eqref{Opt2} in Condition 2 of {Step 2} guarantees the increase of $P_m$, i.e., when $P_m(\boldsymbol{\varphi} + \Delta \boldsymbol{\varphi} ) < P_m(\boldsymbol{\varphi})$, the following inequality
\begin{align}
P_m(\boldsymbol{\varphi} - \Delta \boldsymbol{\varphi} )> P_m(\boldsymbol{\varphi}) \notag
\end{align}
holds.

To this end, we expand $P_m(\boldsymbol{\varphi}  + \Delta \boldsymbol{\varphi} ) $ using Taylor series of $P_m$ with respect to $\boldsymbol{\varphi}$, i.e.,
\begin{align}
&P_m(\boldsymbol{\varphi}  + \Delta \boldsymbol{\varphi} ) \notag \\
= &P_m(\boldsymbol{\varphi} ) +  \frac{\partial P_m(\boldsymbol{\varphi})}{\partial \boldsymbol{\varphi}^H} \Delta  \boldsymbol{\varphi} + \mathcal{O}(\|\Delta  \boldsymbol{\varphi}\|^2) \quad {\rm as}\;\; \Delta  \boldsymbol{\varphi}\rightarrow 0
\end{align}
Since $\| \Delta \boldsymbol{\varphi} \| $ is small, namely, $\| \Delta \boldsymbol{\varphi} \| \rightarrow 0 $, we adopt the first-order approximation, i.e.,
\begin{align}
P_m(\boldsymbol{\varphi} + \Delta \boldsymbol{\varphi} ) \approx P_m(\boldsymbol{\varphi} ) +  \frac{\partial P_m(\boldsymbol{\varphi})}{\partial \boldsymbol{\varphi}^H} \Delta  \boldsymbol{\varphi}
\end{align}

As   it is reported by the system that $P_m(\boldsymbol{\varphi} + \Delta \boldsymbol{\varphi} ) < P_m(\boldsymbol{\varphi})$ in Condition 2, we have
\begin{align}
\frac{\partial P_m(\boldsymbol{\varphi} )}{\partial \boldsymbol{\varphi}^H} \Delta  \boldsymbol{\varphi} < 0
\end{align}
Then, it is easy to verify that
\begin{align}
P_m(\boldsymbol{\varphi}  - \Delta \boldsymbol{\varphi})&\approx P_m(\boldsymbol{\varphi}) -  \frac{\partial P_m(\boldsymbol{\varphi})}{\partial \boldsymbol{\varphi}^H} \Delta  \boldsymbol{\varphi} \notag \\
& > P_m(\boldsymbol{\varphi})
\end{align}
\end{proof}


\begin{remark} {\rm
As $\boldsymbol{\theta}[n] =  e^{-j \pi \boldsymbol{\varphi}[n]}$, the operations of \eqref{Opt1} and \eqref{Opt2} can be written w.r.t. $\boldsymbol{\theta}$ as follows }
\begin{subequations}
\begin{align}
\boldsymbol{\theta} & \leftarrow \boldsymbol{\theta} \odot \Delta \boldsymbol{\theta}  \\
\boldsymbol{\theta}  & \leftarrow  \boldsymbol{\theta} \odot  \Delta \boldsymbol{\theta}^*
\end{align}
\end{subequations}
\end{remark}

\subsection{Comparison with Gradient Ascent Search}

To obtain further insights into the proposed dither-based method, we make a comparison with the well-known iterative algorithm -- gradient ascent (descent in minimization problems) search.

For gradient ascent search algorithm, $\boldsymbol{\varphi}$ in each iteration is updated as follows.
\begin{align}
\boldsymbol{\varphi}  \leftarrow  \boldsymbol{\varphi}  +  \underbrace{ \gamma  \frac{\partial P_m(\boldsymbol{\varphi} )}{\partial \boldsymbol{\varphi} }}_{\Delta \boldsymbol{\varphi}}
\end{align}
When the step size $\gamma$ is small, the iteration will almost surely guarantee the increase of $P_m(\boldsymbol{\varphi} )$, because
\begin{align}
P_m(\boldsymbol{\varphi}  + \Delta \boldsymbol{\varphi} ) & \approx P_m(\boldsymbol{\varphi} ) +  \gamma \frac{\partial P_m(\boldsymbol{\varphi} )}{\partial \boldsymbol{\varphi} ^H} \frac{\partial P_m(\boldsymbol{\varphi} )}{\partial \boldsymbol{\varphi}} \notag \\
& = P_m(\boldsymbol{\varphi} ) +  \gamma  \| \frac{\partial P_m(\boldsymbol{\varphi} )}{\partial \boldsymbol{\varphi}}\|_2^2 \geq P_m(\boldsymbol{\varphi} )
\end{align}

\begin{remark} {\rm
As gradient ascent search take steps in the direction of the gradient, it is also called steepest descent. Thus, the convergence rate of  gradient ascent search  is faster than dither-based extremum search when the step size $\gamma$ is properly selected. On the other hand, it is also noteworthy that gradient ascent search requires the exact expression of the gradient $\frac{\partial P_m(\boldsymbol{\varphi} )}{\partial \boldsymbol{\varphi} }$, whilst dither-based method is implemented through trial and error, which does not rely on any explicit knowledge of the wireless system's internal working mechanism.  }
\end{remark}

\subsection{Integrating ESC Into DRL}

\begin{figure}[tp]{
\begin{center}{\includegraphics[ height=4.2cm]{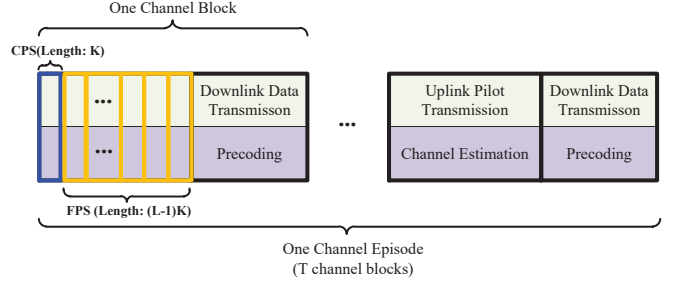}}
\caption{The upgraded frame structure for the integration of ESC into DRL}\label{UpgradedFS1}
\end{center}}
\end{figure}

Recall that the action space for DRL is intentionally restrained for a fast convergence rate, whilst dither-based iterative method relies on the small-scale phase shift. Therefore,  dither-based iterative method is  complementary to the action in DRL and can be applied to enhance the action of DRL.

\vspace{0.05cm}
The \textbf{enhanced action}  of DRL is defined as follows.

\vspace{0.05cm}
\noindent\underline{\emph{Step 1. Coarse phase shift (CPS)}}
\vspace{0.1cm}
When $l = 0$, set
\begin{align}
 \boldsymbol{\theta}^{(t+1)}_{temp} = \boldsymbol{\theta}^{(t)} \odot \Delta \boldsymbol{\theta}_c^{(t)} \label{RefPatternCoarse}
\end{align}
where  $\boldsymbol{\theta}^{(t)}$ is the reflection pattern at the $t$-th channel block, $\Delta \boldsymbol{\theta}_c^{(t)} $ is the coarse incremental phase shift at the $t$-th channel block, and $ \boldsymbol{\theta}^{(t+1)}_{temp} $ is the intermediate reflection pattern at the $t-1$-th channel block.
Following the example in Section III. C,  the action set of the incremental phase shift  $\Delta \boldsymbol{\theta}_c^{(t)} $ is  $\mathcal{A} = \left\{\mathbf{v}(-\frac{6}{N_R}), \mathbf{v}(-\frac{2}{N_R}), \mathbf{v}(0), \mathbf{v}(\frac{2}{N_R}), \mathbf{v}(\frac{6}{N_R})\right\}$.

\vspace{0.15cm}
\noindent\underline{\emph{Step 2. Fine phase shift (FPS)}}
\vspace{0.1cm}\\
For $l$ from $1$ to $L$, do
\begin{align}
\boldsymbol{\theta}^{(t+1)}_{temp}  \leftarrow  \boldsymbol{\theta}^{(t+1)}_{temp} \odot \Delta \boldsymbol{\theta}_f \label{RefPatternFine}
\end{align}
where $\Delta \boldsymbol{\theta}_f = \Delta \boldsymbol{\theta} \; ({\rm or} \; \Delta \boldsymbol{\theta}_f =  \Delta \boldsymbol{\theta}^*)$ is the ascent direction, and $\Delta \boldsymbol{\theta}$ is the random dither signal.

\begin{remark} {\rm
For example, when the quantization level $r=8$, and $N_R=32$, the step of coarse phase shift is $\frac{2\pi}{32}$, and the step of fine phase shift is $\frac{2\pi}{256}$ with the range $[-\frac{\pi}{32} , \; \frac{\pi}{32}]$. }
\end{remark}

To be compatible with the enhanced action, the frame structure needs to be updated as in \figref{UpgradedFS1}. In the first $K$ time slots, UEs transmit pilots and BS performs channel estimation with the reflection pattern in  \eqref{RefPatternCoarse}, and in the subsequent $(L-1)K$ time slots, UEs repeatedly transmit pilots and BS performs channel estimation, while the reflection pattern updates as in  \eqref{RefPatternFine}. It is noteworthy that, as the  performance feedback is done once per channel block, the performance metric used for the dither-based method is an approximation  derived by replacing the authentic channel response $\mathbf{H}$ in  \eqref{Performance} with the channel estimate $\hat{\mathbf{H}}$.

\begin{remark}{\rm
The parameter $L$ can be set adaptively according to the channel dynamics. For a practical wireless communication system, different values of $L$ correspond to different modes.
}
\end{remark}

\section{Numerical Results}

In this section, we present some numerical results to verify the effectiveness of our proposed model-free control of IRS \footnotemark.

\footnotetext{The simulation code is available at https://github.com/WeiWang-WYS/IRSconfigurationDRL}

\subsection{Simulation Parameters}

The BS is equipped with a ULA that is placed along the direction $[1, 0, 0]$ (i.e., x-axis), and IRS is a ULA, which is placed  along the direction $[0, 1, 0]$ (i.e., y-axis), UEs are equipped with a single antenna,  the user number is $K=2$, BS antenna number is $N_B = 2$, and IRS reflector number if $N_R = 32$. The element antennas/reflectors of BS and IRS are both with half wavelength spacing.  The position of BS is $[0, 0, 10]$, the position of IRS is $[-2, 5, 5]$, and the UEs are uniformly distributed in the area $[0, 10) \times [0, 10)$  with the height being $1.5$. The noise variance at BS side is $\sigma_B^2 = 0.1$,  the noise variance at UE side is $\sigma^2_k = 0.5, \; \forall k\in \{1,\cdots,K \}$. Each channel episode consists of $20$ channel blocks. {In each channel episode, the LoS component is generated by randomly selecting the user locations within the area $[0, 10) \times [0, 10)$, and the LoS component is time-invariant within the $20$ channel blocks of that channel episode.}

The LoS channel between BS and IRS is
\begin{align}
\mathbf{H}_{BR, LoS} =  \mathbf{v}_R  \mathbf{v}_B^H
\end{align}
where the steering vectors are represented as
\begin{subequations}
\begin{align}
\mathbf{v}_R  = \mathbf{v}(\Psi_R, N_{R,y}) &= \left[1,\; e^{j  \pi \Psi_R},\; \cdots,\; e^{j  (N_{R,y} - 1)\pi \Psi_R} \right]^T \notag \\
\mathbf{v}_B = \mathbf{v}(\Psi_B, N_{B,x}) &= \left[1,\; e^{j  \pi \Psi_B},\; \cdots,\; e^{j  (N_{B,x} - 1)\pi \Psi_B} \right]^T \notag
\end{align}
\end{subequations}
and, according to  \cite{JitteringWang}, the directional cosines $\Psi_R,  \Psi_B$ are given by
\begin{subequations}
\begin{align}
\Psi_R & = [0, 1, 0] \mathbf{e}_{BR} = \mathbf{e}_{BR}(2) \\
\Psi_B  &   = [1, 0, 0] \mathbf{e}_{BR} = \mathbf{e}_{BR}(1)
\end{align} \label{BSside}
\end{subequations}
where the direction vector $\mathbf{e}_{BR}$ is determined by the relative position of BS and UE, i.e.,
\begin{align}
\mathbf{e}_{BR} &\triangleq \frac{\mathbf{p}_{B} - \mathbf{p}_{R}}{\|\mathbf{p}_{B} - \mathbf{p}_{R}\|_2} \label{DirectionEq}
\end{align}
The NLoS components are Gaussian distributed, i.e., $\mathbf{H}_{BR, NLoS}(\ell, \kappa) \in \mathcal{CN}(0, 1)$. The channel matrix $\mathbf{H}_{BU}$ and $\mathbf{H}_{RU}$ are generated in the same way.

\begin{figure}[t]
\begin{minipage}[!h]{.48\linewidth}
\centering
\includegraphics[width=1.1\textwidth]{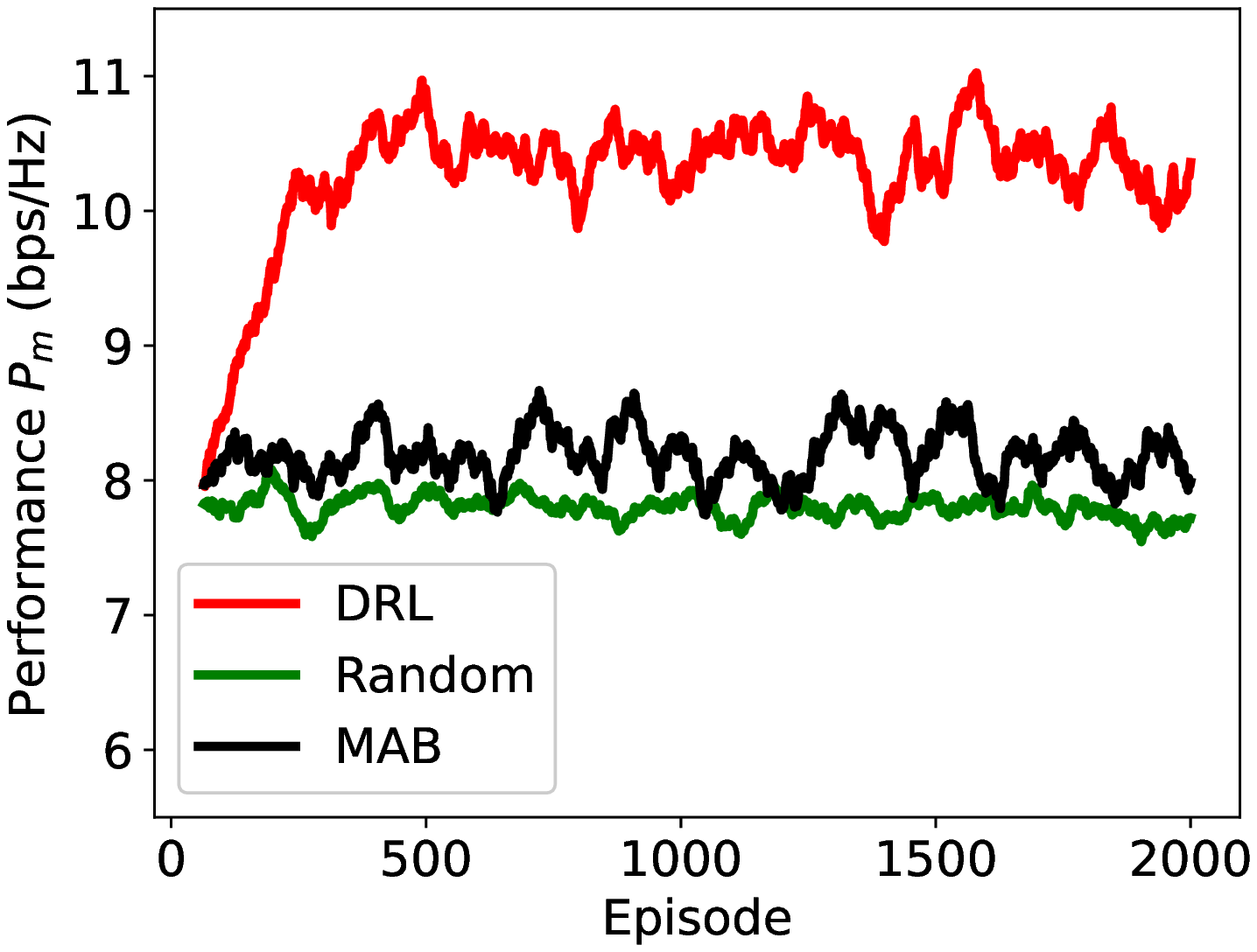}
\subcaption{$K_{Rician}=5$}
\label{SpectralEfficiencySNR1}
\end{minipage}
\begin{minipage}[!h]{.48\linewidth}
\centering
\includegraphics[width=1.1\textwidth]{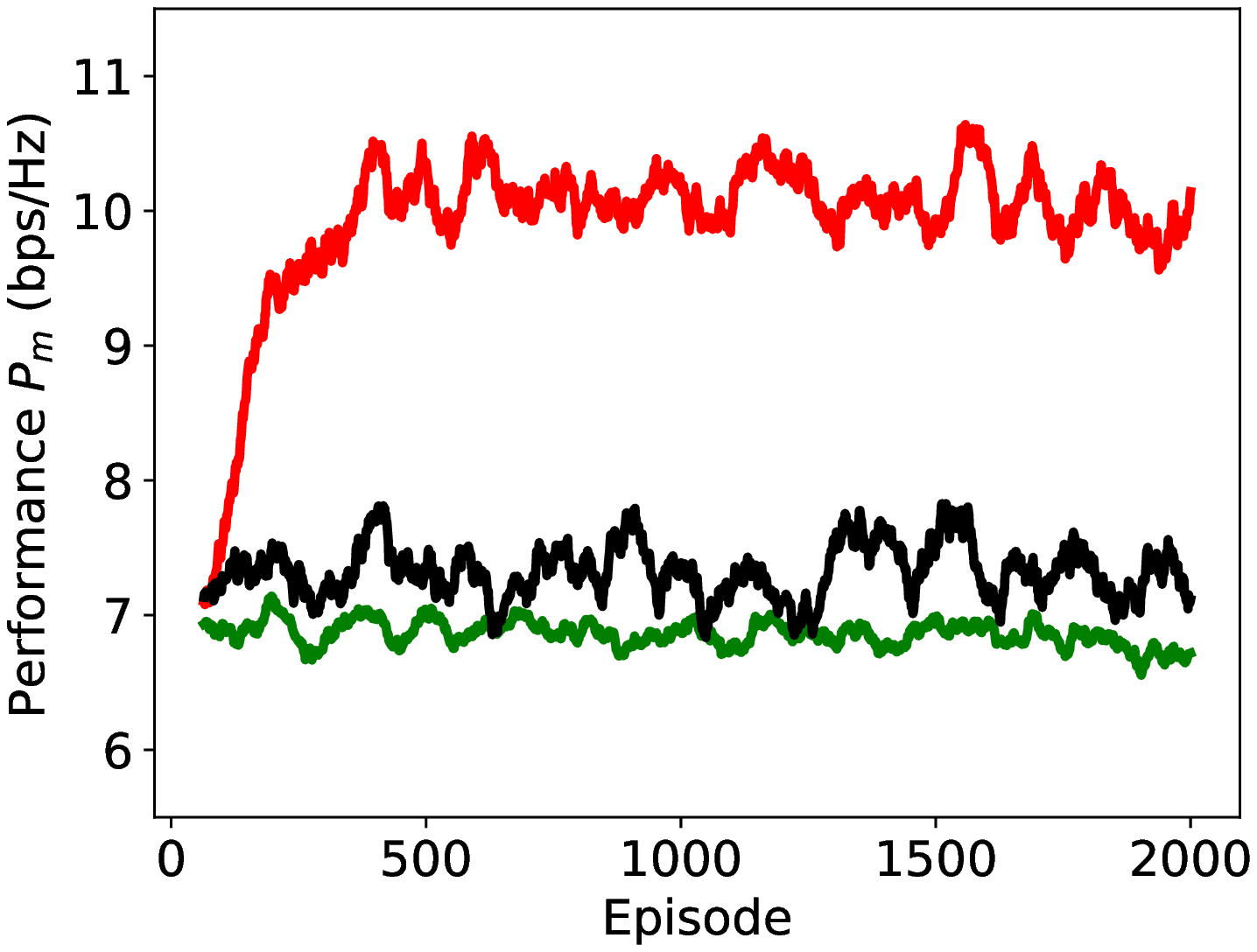}
\subcaption{$K_{Rician}=10$}
\label{SpectralEfficiencySNR6}
\end{minipage}
\begin{minipage}[!h]{.48\linewidth}
\centering
\includegraphics[width=1.1\textwidth]{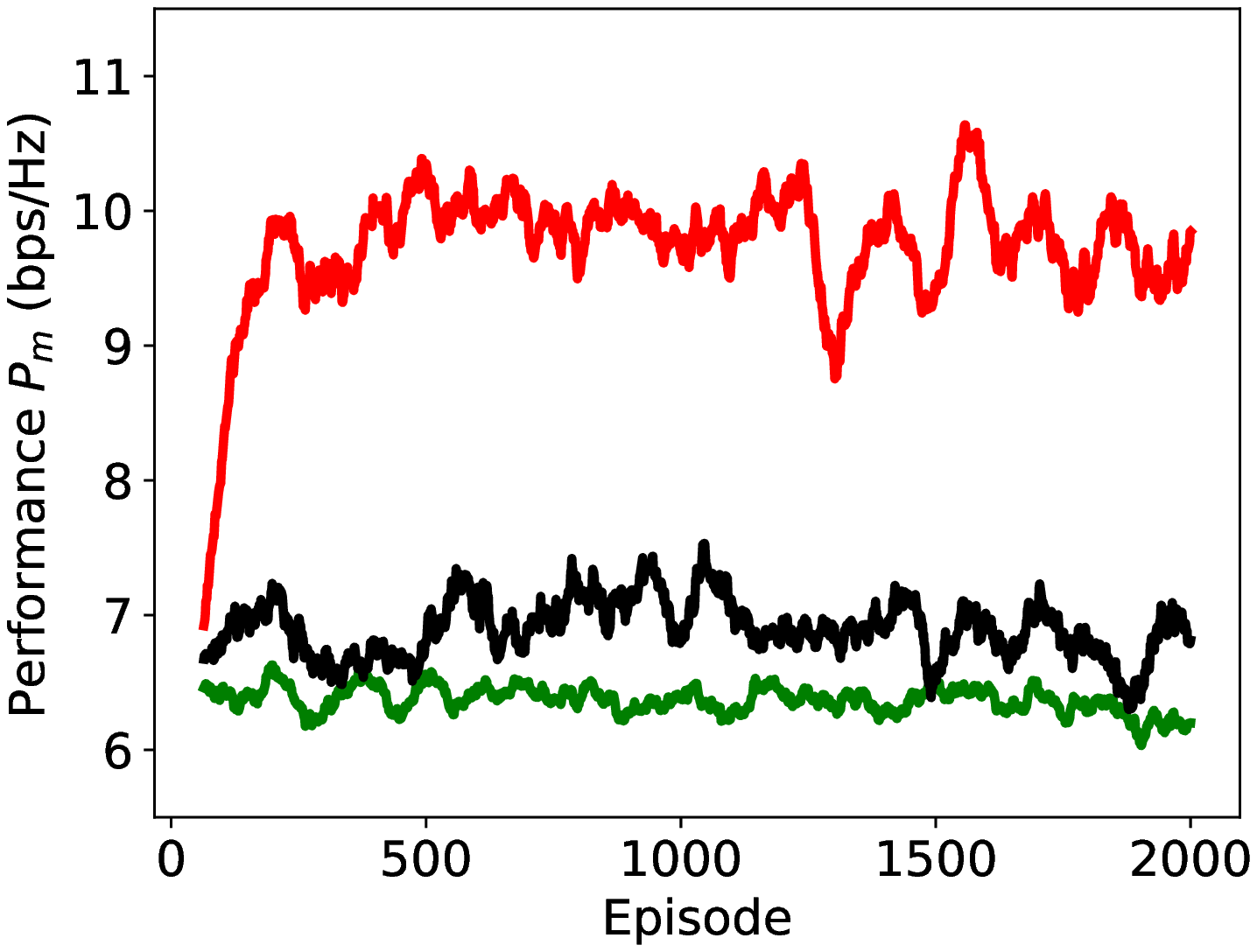}
\subcaption{$K_{Rician}=15$}
\label{SpectralEfficiencySNR11}
\end{minipage}
\hspace{0.25cm}
\begin{minipage}[!h]{.48\linewidth}
\centering
\includegraphics[width=1.1\textwidth]{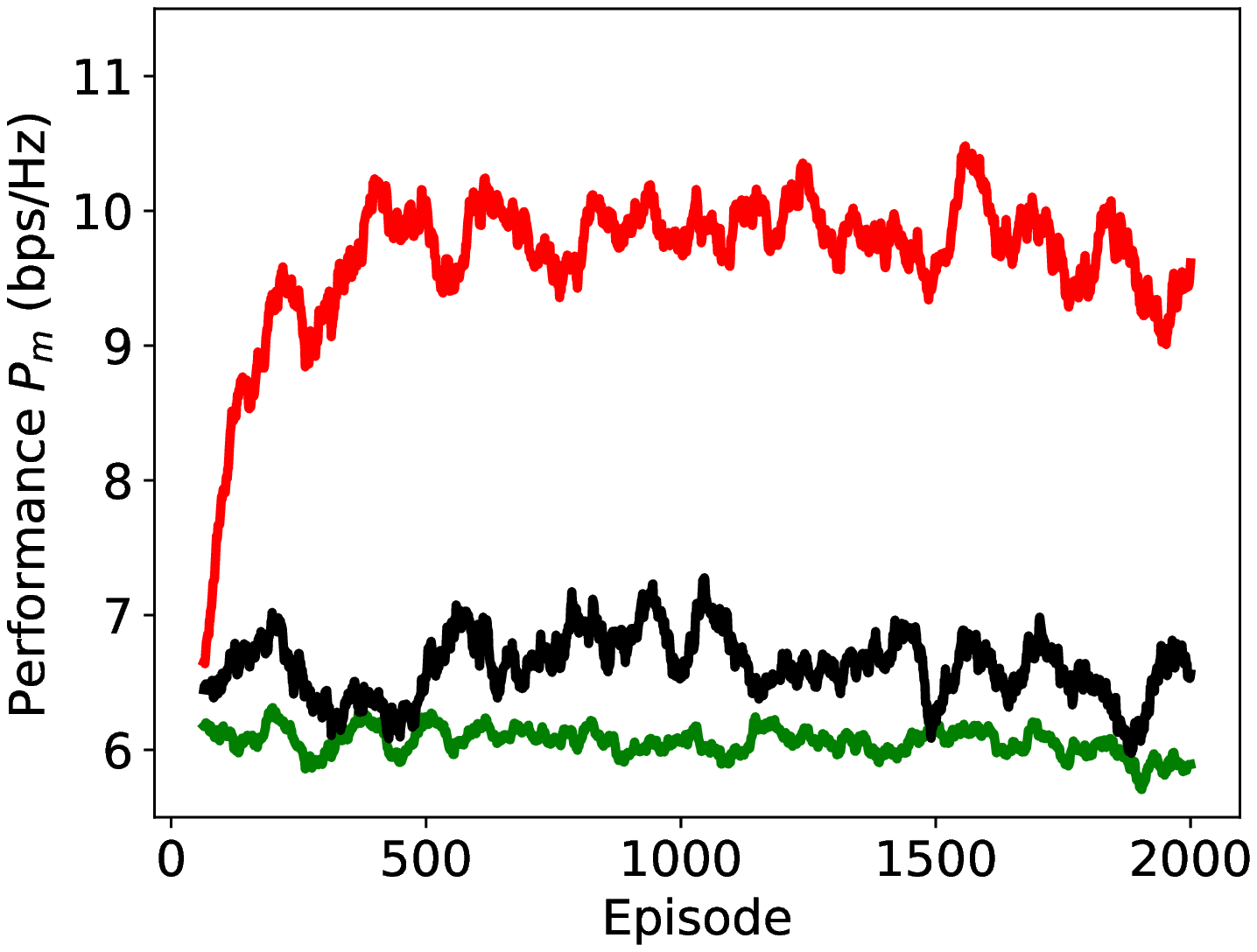}
\subcaption{$K_{Rician}=20$}
\label{SpectralEfficiencySNR16}
\end{minipage}
\caption{Moving average of $P_m$ for DRL under different values of Rician factor} \label{RicianFactor}
\end{figure}

\subsection{Performance Study of DRL}

\begin{figure}[t]
\begin{minipage}[!h]{1\linewidth}
\centering
\includegraphics[width=0.9\textwidth]{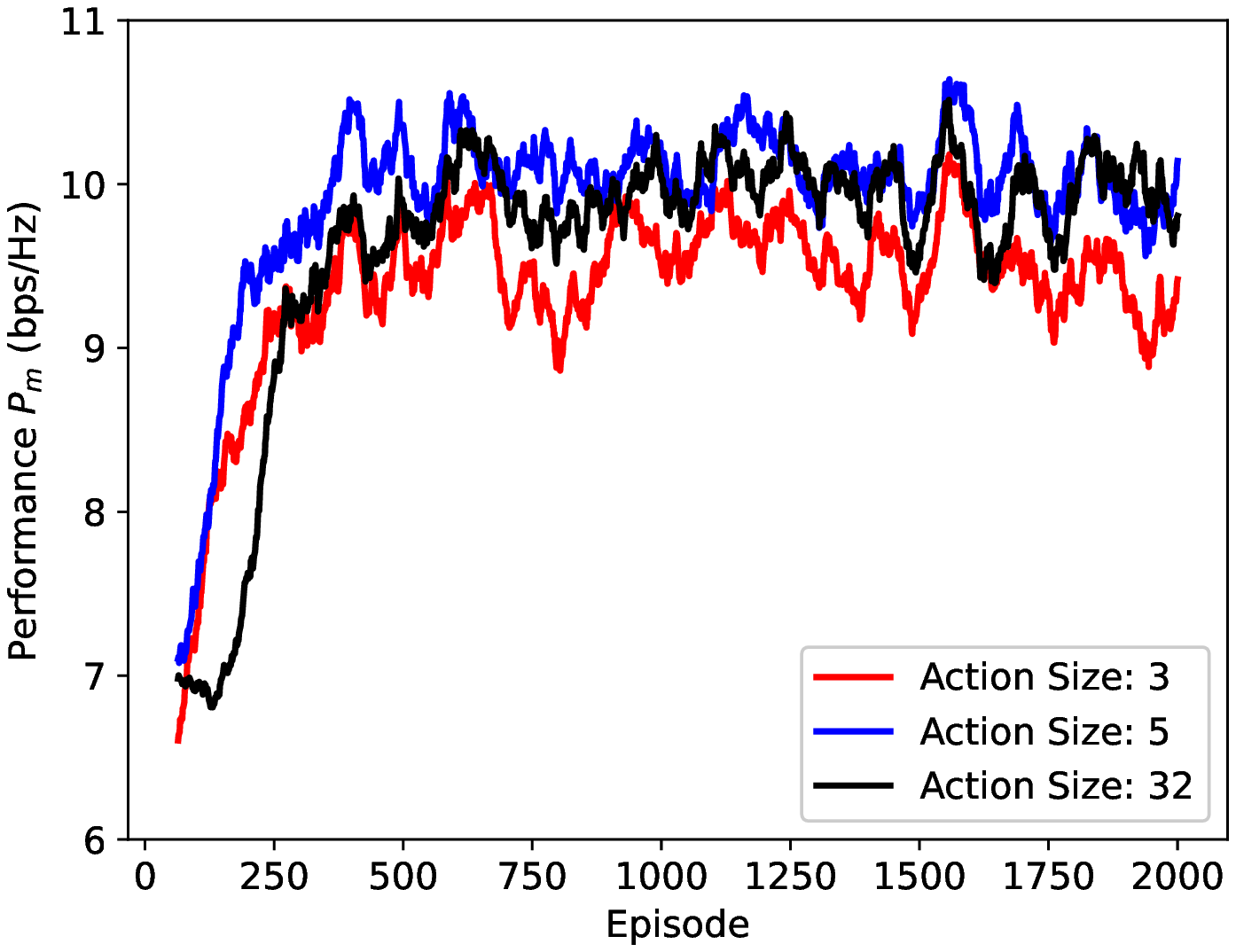}
\subcaption{Moving average of $P_m$ for DRL }
\label{ActionComp}
\end{minipage}
\begin{minipage}[!h]{1\linewidth}
\centering
\includegraphics[width=0.9\textwidth]{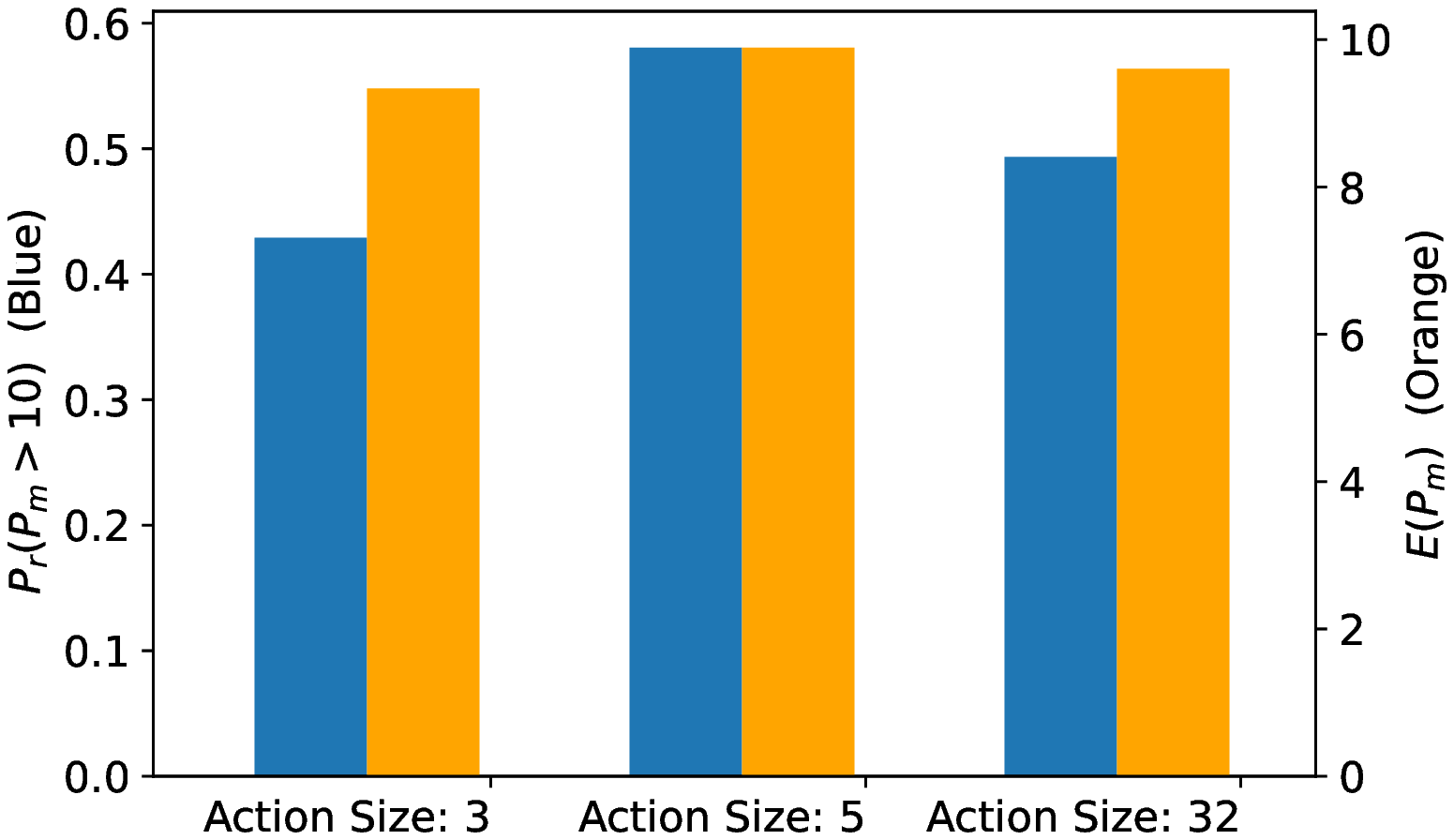}
\vspace{-.8cm}
\subcaption{Average sum rate $\mathbb{E}(P_m)$ (orange color) and the probability $P_r(P_m > 10)$ (blue color)}
\label{AveOutage}
\end{minipage}
\caption{Performance  of DRL with different action sets } \label{ActionSize}
\end{figure}

In \figref{RicianFactor}, we study the performance of the proposed DRL scheme (with the action set $\mathcal{A}_5$) under different values of Rician factor. The x-axis represents the episode, and the y-axis represents the moving average of $P_m$ (window length is $64$).  As can be seen, the proposed DRL significant outperforms the benchmark schemes, i.e., random reflection and multi-armed bandit (MAB). Different from DRL, the actions of random reflection and MAB that we used are absolute phase shift, and the action set is the DFT vectors. Although all the three schemes are independent of the sub-channel CSI, their utilizations of the other information are different. Random reflection is independent of any information and undoubtedly  achieves the worst performance; MAB assumes a fixed distribution of \emph{rewards} and explores the reward distributions of all arms. However, MAB fails to describe the state of the environment and to build the connection between the action and the environment; DRL defines an appropriate \emph{state} to represent the agent's ``position" within the environment and learns the quality of a state-action combination using the DQN from the information of \emph{rewards} and \emph{states}, which enables the agent to choose the best action to maximize the returns. We can also find that the performance gap of DRL and the benchmarks schemes becomes larger with the increase of Rician factor $K_{Rician}$, which indicates that  the effectiveness of DRL is also dependent on the radio environment.

In \figref{ActionSize}, we study the impacts of action size to the performance of the proposed DRL scheme when the Rician factor is $K=10$. In addition to the action set  $\mathcal{A}_5 = \left\{\mathbf{v}(-\frac{6}{N_R}), \mathbf{v}(-\frac{2}{N_R}), \mathbf{v}(0), \mathbf{v}(\frac{2}{N_R}), \mathbf{v}(\frac{6}{N_R})\right\}$  defined in Section III, we adopt the action set  $\mathcal{A}_3 = \left\{ \mathbf{v}(-\frac{2}{N_R}), \mathbf{v}(0), \mathbf{v}(\frac{2}{N_R})\right\}$ and action set $\mathcal{A}_{32} = \left\{ \mathbf{v}(-1),  \mathbf{v}(-1-\frac{2}{N_R}), \cdots, \mathbf{v}(1-\frac{2}{N_R})\right\}$ (namely, DFT matrix) as the benchmarks. From \figref{ActionComp}, we can see that $\mathcal{A}_5$ is the fastest to converge, while $\mathcal{A}_{32}$ is the slowest.  Although a large action size will speed up the response rate of the agent, it will, on the other hand, demand more time for the DQN to converge. In the convergence region, we find that $\mathcal{A}_5$ and $\mathcal{A}_{32}$ achieve the similar performance, while $\mathcal{A}_{3}$'s performance is inferior.  It indicates that a well-design action set with a moderate size might be better than the small action size and the over-large action size. In \figref{AveOutage}, the average sum rate $\mathbb{E}(P_m)$, and the probability $P_r(P_m >10)$ are presented in the bar chart. For  $\mathcal{A}_5$,   $\mathbb{E}(P_m) = 9.89$ bps/Hz and $P_r(P_m>10) = 58.05\%$; for $\mathcal{A}_3$, $\mathbb{E}(P_m) = 9.34$ bps/Hz and $P_r(P_m>10) = 42.9\%$;  for $\mathcal{A}_{32}$, $\mathbb{E}(P_m) = 9.60$ bps/Hz and $P_r(P_m>10) = 49.35\%$. It further verifies the importance of action set design.

\subsection{Performance Study of ESC}

\begin{figure}[tp]{
\begin{center}{\includegraphics[ height= 6cm]{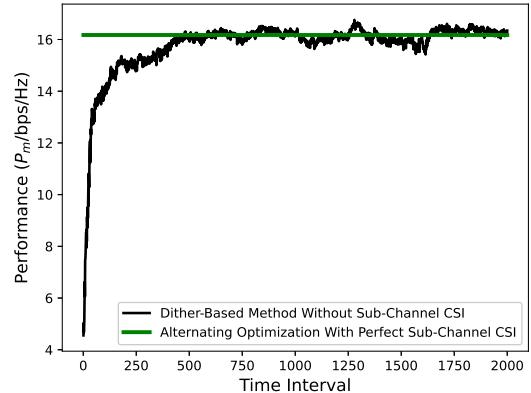}}
\caption{Performance of ESC-inspired dither-based iterative method}\label{ESCComparison10}
\end{center}}
\end{figure}

In \figref{ESCComparison10}, we study the performance of the ESC-inspired dither-based iterative method in a specific channel block. Each time interval of $x$-axis consists of $K$ time slots, which is the pilot length required by the BS to estimate the multi-user channel $\mathbf{H}$.  As can be seen that the performance metric $P_m$ for dither-based method is almost monotonically increasing over time. Note that the performance metric used to guide the dither-based iterative method is an approximation, rather than the authentic feedback. Thus, the slight fluctuation of the performance curve is reasonable. For the purpose of comparison, we adopt the model-based method as the benchmark, in which the perfect sub-channel CSI, namely, $\mathbf{H}_{BU}$, $\mathbf{H}_{BR}$, and $\mathbf{H}_{RU}$, is available. The optimal reflection coefficient vector $\boldsymbol{\theta}$ is derived through solving the optimization problem \eqref{Opt}. Recall that the difference from model-free control is that the exact relationship between the objective function $P_m$ and the variable $\boldsymbol{\theta}$ is known in model-based methods. Due to the discrete nature of the feasible region $\mathcal{B}$, the optimization problem is intractable. Hence, we manage to solve it using the alternating optimization technique, which alternatively freezes $N_R-1$ reflection coefficients and optimize only $1$ reflection coefficient. According to \figref{ESCComparison10}, the model-free dither-based method achieves almost the same performance as the model-based alternating optimization when the time index is greater than $500$.  However, in practice, we have to weigh the cost of time resources against the benefits. As the dither-based method needs to sample $\hat{\mathbf{H}}$, one iteration means the cost of a unit of time resource in wireless communications. Thus, in order to balance the time allocation between pilot transmission and data transmission, the time resources dedicated to dither-based method (i.e., the time for pilot transmission) should be deliberately selected according to channel dynamics (i.e., the length of a channel block).

\subsection{Performance Study of the Integrated DRL and ESC}

\begin{figure}[t]
\begin{minipage}[!h]{.48\linewidth}
\centering
\includegraphics[width=1.1\textwidth]{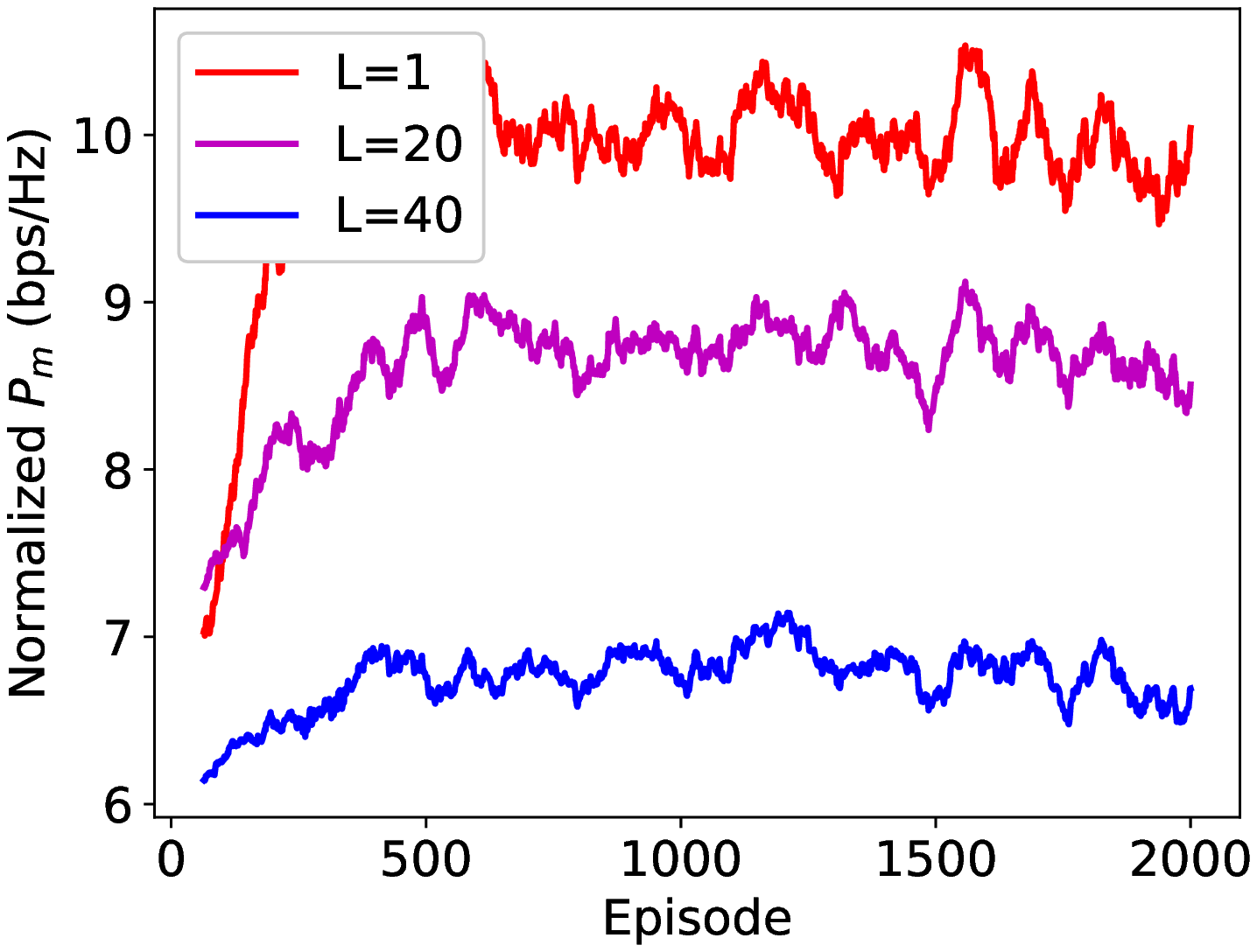}
\subcaption{ $L_B = 100$}
\label{SpectralEfficiencySNR1}
\end{minipage}
\begin{minipage}[!h]{.48\linewidth}
\centering
\includegraphics[width=1.1\textwidth]{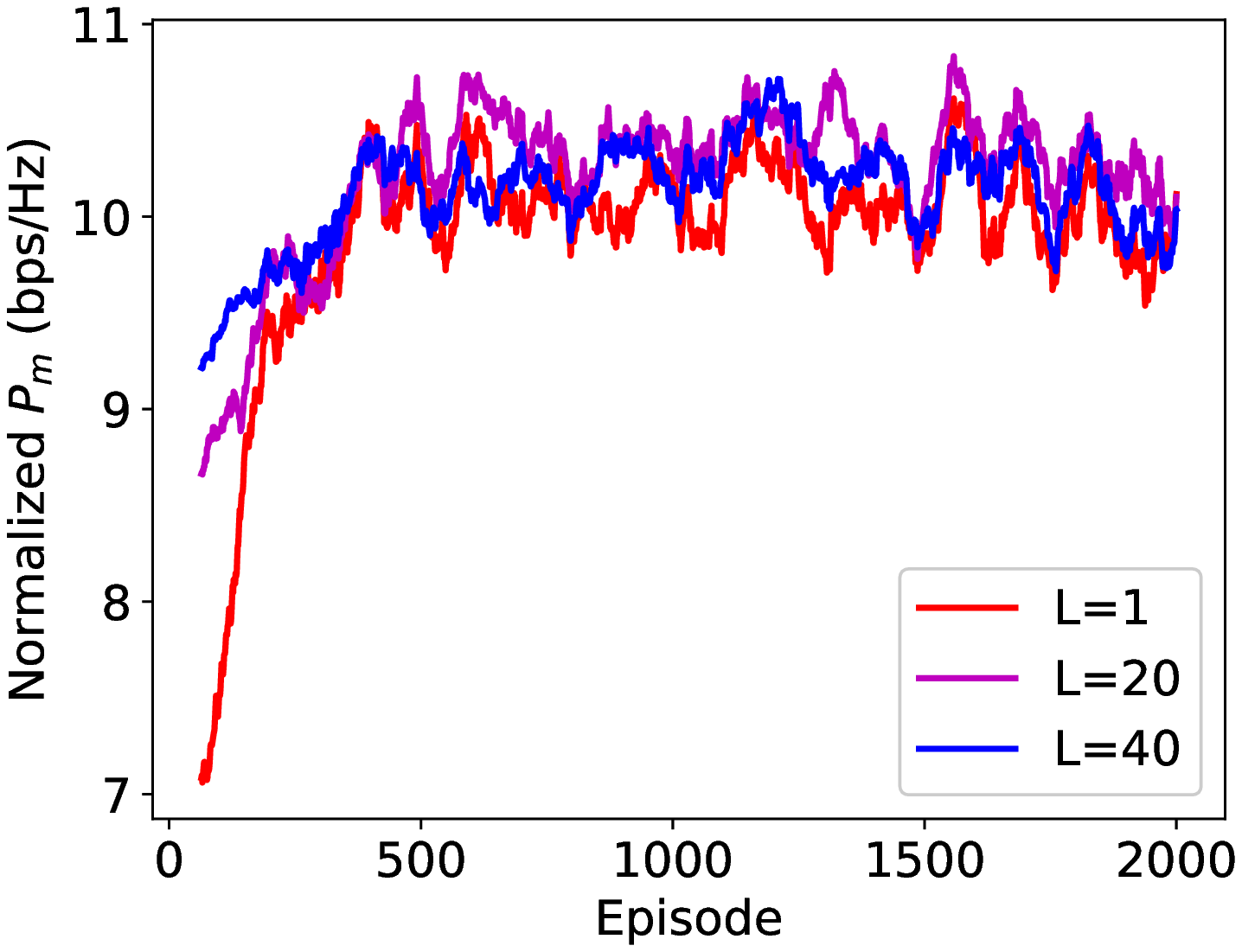}
\subcaption{ $L_B = 400$}
\label{SpectralEfficiencySNR6}
\end{minipage}
\begin{minipage}[!h]{.48\linewidth}
\centering
\includegraphics[width=1.1\textwidth]{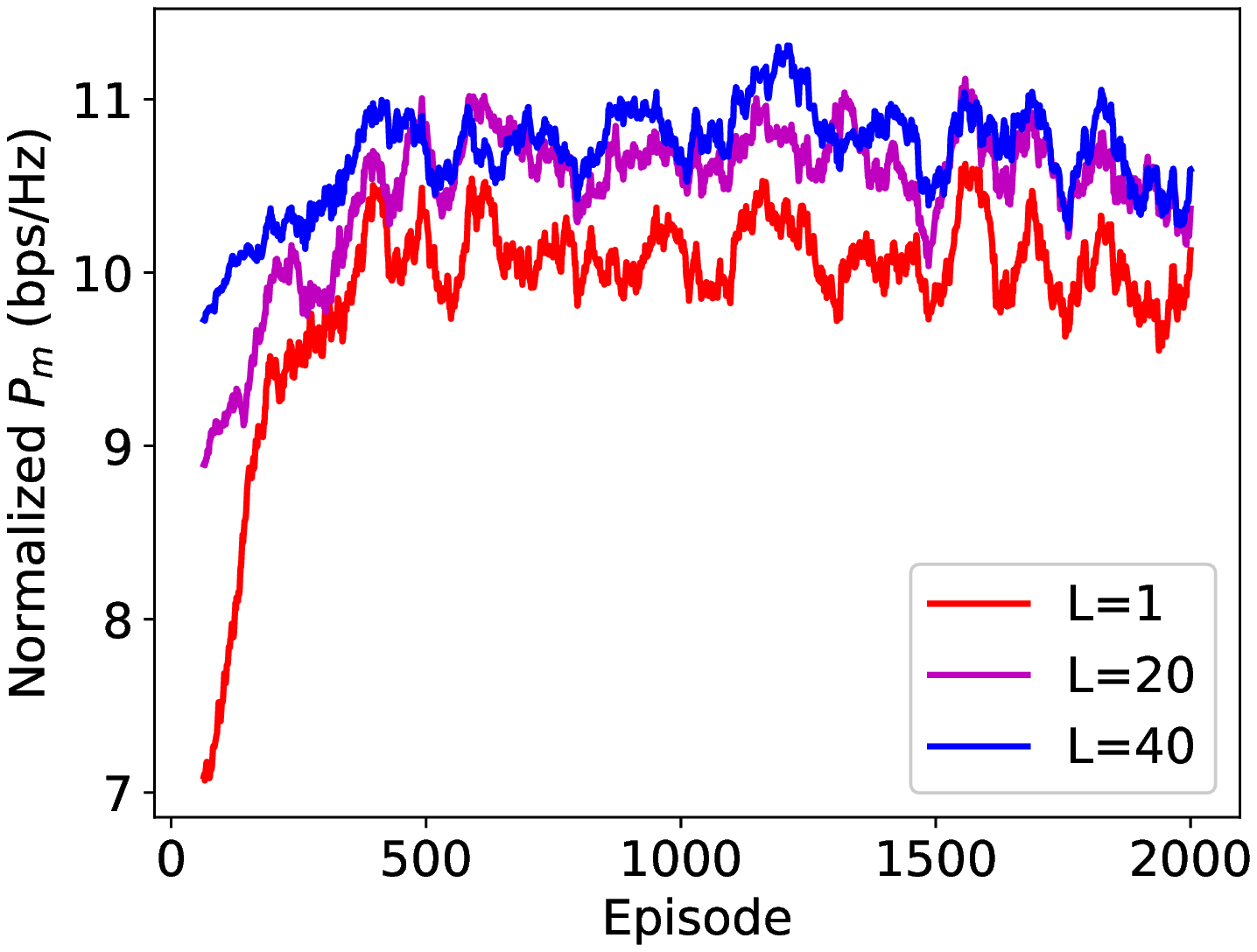}
\subcaption{$L_B = 800$}
\label{SpectralEfficiencySNR11}
\end{minipage}
\hspace{0.25cm}
\begin{minipage}[!h]{.48\linewidth}
\centering
\includegraphics[width=1.1\textwidth]{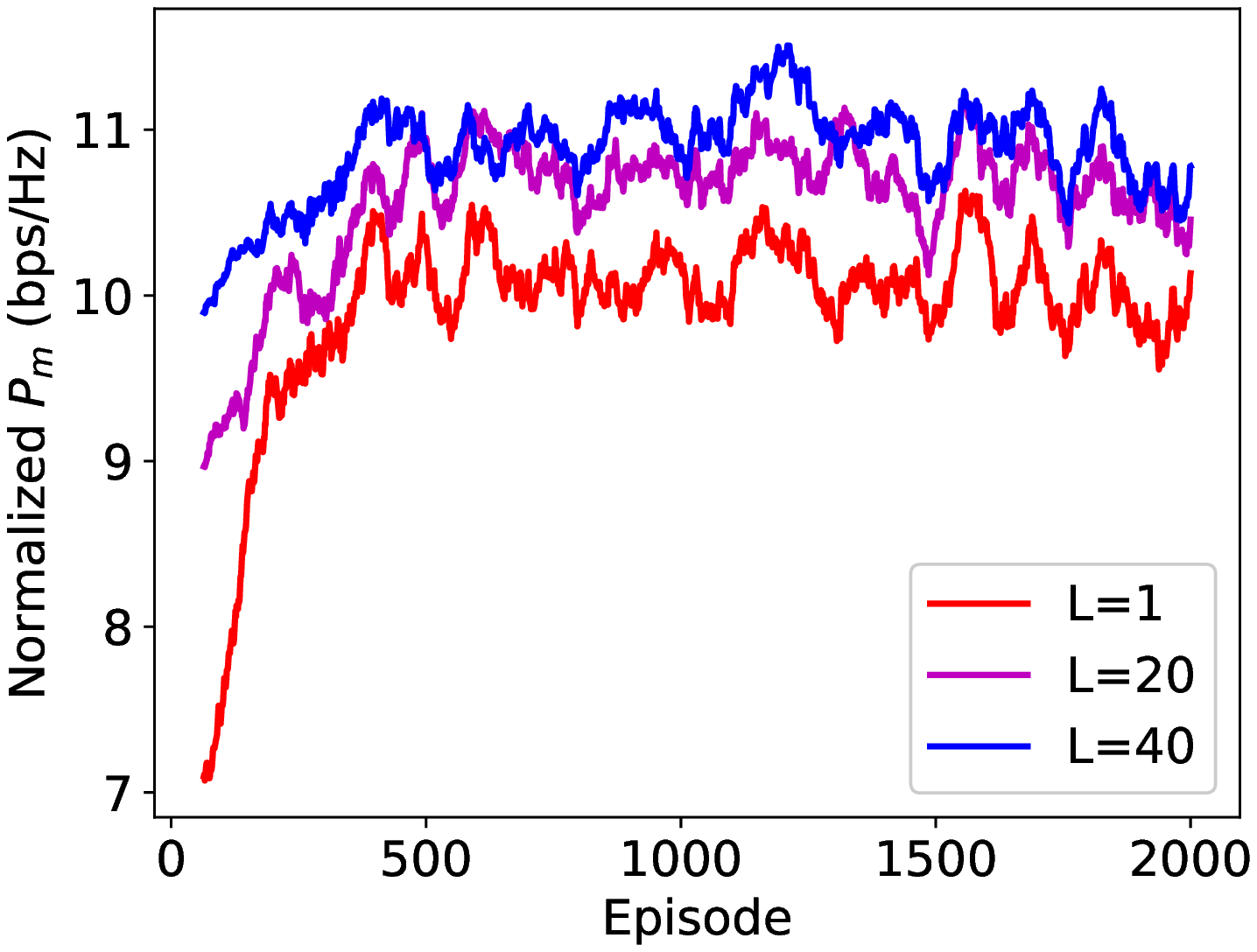}
\subcaption{$L_B = 1200$}
\label{SpectralEfficiencySNR16}
\end{minipage}
\caption{Moving average of the normalized $P_m$ for the integrated DRL and ESC in different channel dynamics} \label{IntegratedDRLESC}
\end{figure}

In \figref{IntegratedDRLESC}, we study the performance of the integrated DRL and ESC method  (with the action set $\mathcal{A}_5$) in different channel dynamics. The normalized  $P_m$ is the obtained by multiplying $P_m$ by the coefficient $\frac{L_B-L}{L_B}$, where $L_B$ is the channel block length and $L$ is the training length. Take $L_B = 100$ and $L=1$ as an example, the first $L=1$ time interval is used for pilot transmission and the rest $L_B-L=99$ time intervals are used for data transmission.
{It is also noteworthy that when $L=1$ the scheme is DRL only, and when $L\geq 2$ the scheme is the integrated DRL and ESC.}
In addition, we adopt the action set  $\mathcal{A}_5$.
From the figure, we can see that when the channel block length $L_B=100$, DRL outperforms the integrated DRL and ESC with $L=20, 40$. However,  when the channel block length increases, the integrated DRL and ESC gradually becomes superior, which verifies its effectiveness in the slow fading channel.  Therefore, the parameter $L$ of the integrated DRL and ESC can be set adaptively to accommodate different channel dynamics.


\section{Conclusion}

In this paper, we have proposed a model-free control of IRS that is independent of sub-channel CSI.  We firstly model the control of IRS as an MDP and  apply DRL to perform real-time coarse phase control of IRS. Then, we apply ESC as the fine phase control of IRS.  Finally, by updating the frame structure, we integrate DRL and ESC in the model-free control of IRS to improve its adaptivity to different channel dynamics. Numerical results show the superiority of our proposed  scheme in model-free IRS control without sub-channel CSI.

\bibliographystyle{IEEEtran}%

\bibliography{bibfile}

\end{document}